\documentclass[11pt]{article}
\usepackage{graphicx}
\usepackage{amsmath,amsthm}
\usepackage{amssymb}
\usepackage{fullpage}
\usepackage{paralist}
\input{Qcircuit}

\newcommand{\abs}[1]{\left\vert#1\right\vert}
\newcommand{\set}[1]{\left\{#1\right\}}
\newcommand{\complex}{{\mathbb C}}
\newcommand{\reals}{{\mathbb R}}
\newcommand{\polylog}{\mbox{polylog\,}}
\newcommand{\poly}{\mbox{poly\,}}
\newcommand{\Tr}{{\rm Tr}}
\def\ket#1{| #1 \rangle}
\def\bra#1{\langle #1 |}
\def\bracket#1#2{\langle #1 | #2 \rangle}
\def\ketbra#1#2{| #1 \rangle\!\langle #2 |}
\newcommand{\LL}{\mathcal{L}}

\newcommand{\spa}[1]{\mathcal{#1}}
\newcommand{\norm}[1]{\left\|\,#1\,\right\|}       
\newcommand{\trnorm}[1]{\norm{#1}_{\mathrm {tr}}}  
\newcommand{\snorm}[1]{\norm{#1}}    
\newcommand{\comp}{{\rm c}}
\newcommand{\unco}{{\rm u}}
\newcommand{\either}{\eta}
\newcommand{\instate}[1]{\psi_{#1}}
\newcommand{\instalo}[1]{\widetilde\psi_{#1}}
\newcommand{\midstate}[2]{\psi_{#1,#2}^{(d_1,\ldots,d_{#2})}}
\newcommand{\midrho}[2]{\rho_{#1,#2}}
\newcommand{\chan}[2]{{\cal E}_{#1,#2}}
\newcommand{\chant}[1]{{\cal E}_{#1}}
\newcommand{\measp}[2]{M_{#1,d_{#2},{\rm proj}}^{(d_1,\ldots,d_{{#2}-1})}}
\newcommand{\meas}[2]{M_{#1,d_{#2}}^{(d_1,\ldots,d_{{#2}-1})}}
\newcommand{\measc}[1]{M'_{#1,{\bf d}}}
\newcommand{\measb}[1]{M_{#1,{\bf b}}}
\newcommand{\measo}[2]{M_{#1,#2}}
\newcommand{\pr}[1]{p_{\bf b}^{#1}}
\newcommand{\row}[1]{\rho_{\bf b}^{#1}}
\newcommand{\trc}{\Tr_{\rm ctrl}}
\newcommand{\co}{\gamma}
\newcommand{\ca}{\xi}
\newcommand{\s}{s}
\newtheorem{theorem}{Theorem}[section]
\newtheorem{lemma}[theorem]{Lemma}
\newtheorem{definition}[theorem]{Definition}
\newtheorem{algorithm}[theorem]{Algorithm}

\begin{document}

\title{\textbf{Gate-efficient discrete simulations of continuous-time quantum query algorithms%
\thanks{Research supported by Canada's NSERC, CIFAR, MITACS, the U.S. ARO, and ARC grant FT100100761.}}}
\author{Dominic W. Berry\footnote{Department of Physics and Astronomy, Macquarie University, NSW 2109, Australia.}, Richard Cleve\footnote{David R. Cheriton School of Computer Science and Institute for Quantum Computing, University of Waterloo.}, and Sevag Gharibian$^{\ddagger}$\footnote{Electrical Engineering \& Computer Sciences, University of California
at Berkeley, USA.}}

\date{\today}
\maketitle

\begin{abstract}
We show how to efficiently simulate continuous-time quantum query algorithms that run in time
$T$ in a manner that preserves the query complexity (within a polylogarithmic factor)
while also incurring a small overhead cost in the total number of gates between queries.
By small overhead, we mean $T$ within a factor that is polylogarithmic in terms of $T$
and a cost measure that reflects the cost of computing the driving Hamiltonian.
This permits \textit{any} continuous-time quantum algorithm based on an efficiently
computable driving Hamiltonian to be converted into a gate-efficient algorithm with
similar running time.
\end{abstract}


\section{Introduction and Summary of Result}

The standard quantum query model can be represented as an oracle that performs the unitary operation
$\ket{j,k}\mapsto\ket{j,k \oplus x_j}$, where $x_1 x_2 \dots x_L \in \{0,1\}^L$ is the data, and $\oplus$ indicates modular addition.
A convenient representation of the oracle is given by removing the ancilla, and having the oracle give a phase shift,
so the unitary operation for the oracle, $Q$, acts as $Q\ket{j}=(-1)^{x_j}\ket{j}$.
The fractional query model is a natural variant of this, where the operation is $Q^\lambda\ket{j}=(-1)^{\lambda x_j}\ket{j}= e^{i \pi \lambda x_j} \ket{j}$, and $\lambda$
may be taken to be arbitrarily small (but positive).
In the fractional query model, each size-$\lambda$ query is taken to have cost $\lambda$.
The fractional query model potentially provides more power than the standard query model,
because additional unitary operations (which are independent of $x_j$) can be performed in between the fractional queries.

Informally, the continuous-time query model~\cite{FarhiG1998a} arises from the fractional query model in the limit as $\lambda$ approaches zero.
More formally, in the continuous-time query model, the oracle operation is replaced with an oracle Hamiltonian, $H_Q$, which acts as $H_Q\ket{j}=x_j\ket{j}$.
Evolving under this Hamiltonian for time $\pi$ would result in a full discrete query.
The additional operations are replaced with a \emph{driving Hamiltonian} independent of $x_j$, which we denote $H$ ($H$ may be time-dependent).
The algorithm then becomes Hamiltonian evolution with the sum of the oracle and driving Hamiltonians, and
the complexity is quantified by the time of evolution.
The continuous-time and fractional query models are equivalent in the sense that each can simulate the other (to any desired level of accuracy) with the same query cost.
For example, a continuous-time query algorithm can be approximated using fractional queries via a Lie-Trotter formula~\cite{CleveG+2009}.
The continuous-time and discrete query models are also effectively equivalent, in that one can convert from one to the
other with at most a polylogarithmic overhead in the query cost~\cite{CleveG+2009}.

Presently, we are concerned not just with the query cost, but with the cost in terms of the number of additional
gates and ancilla qubits needed.
We show that any continuous-time quantum query algorithm whose total query time
is $T$ and whose driving Hamiltonian is implementable with $G$ 1- and 2-qubit gates (in a sense defined in Section~\ref{sec:statement_main}) can be simulated by a discrete-query quantum algorithm using the
following resources:
\begin{itemize}
\item $O(T \log T / \log\log T)$ queries
\item $O(T G \,\log(T)+T\log^3(\|H\|T))$ 1- and 2-qubit gates [or $O(TG \log(T) + TG^3)$ in terms of just $T$ and $G$]
\item $O(\log^3(\|H\|T))$ qubits of space [or $O(G^3)$].
\end{itemize}
This extends the previous result~\cite{CleveG+2009} where the query cost is the same, but where
the orders of the second and third resource costs are at least $T^2\polylog T$ and $T \polylog T$
respectively.
The present result can also be compared with the result~\cite{LeeM+2011} where the query
cost is superior to ours, $O(T)$ (which is asymptotically optimal), but whose
methodology does not (as far as we know) yield an efficient gate construction from
an efficiently implementable driving Hamiltonian.

Another advantage of our result is that it provides an exponential improvement in the scaling (of the number of gates and ancilla qubits) with $\|H\|$ over that in \cite{CleveG+2009}.
Here the number of gates is polylogarithmic in $\|H\|$, whereas it is superlinear in $\|H\|$ in \cite{CleveG+2009}.
This is important, as the norm of the driving Hamiltonian can potentially be large.

\section{Significance to Quantum Computation}

The continuous-time query model is an important tool for designing algorithms, and for example yielded the algorithm for AND-OR tree evaluation~\cite{FarhiG+2007}.
The difficulty with continuous-time quantum algorithms is that, in order to implement them on quantum computers,
these abstract query algorithms need to be translated into concrete algorithms with subroutines substituted for the black-box queries\footnote{A query is typically \emph{not} something that could be physically implemented directly via continuous-time Hamiltonian evolution, as in an analog quantum computer.
A query corresponds to the coherent evaluation of a classical function on several qubits, and requires several quantum gates to implement, regardless of whether it is a full query or a fractional query.}.
In these circumstances, what matters is the total gate complexity, which can be
large if the cost of the operations performed between the queries is large, even if the
number of queries is small.
The contribution of our result is that it provides a \textit{systematic} way to obtain a gate-efficient
discrete-query algorithm from \textit{any} continuous-time query algorithm where the driving
Hamiltonian can be efficiently implemented.
That is, whenever the implementation cost of the driving Hamiltonian is small, the total gate complexity
is not much more than the query complexity times the cost of implementing each query.

Consider applying the continuous-time quantum algorithm in~\cite{FarhiG+2007} for AND-OR tree evaluation to evaluate expressions of the form
\begin{equation}
\exists x_1 \forall x_2 \exists x_3 \cdots \forall x_L f(x_1,x_2,\dots,x_L),
\end{equation}
where one is given a polynomial (in~$L$) size circuit implementation of
$f: \{0,1\}^L \rightarrow \{0,1\}$.
This corresponds to evaluating a balanced binary AND-OR tree of size $N = 2^L$.
A continuous-time query algorithm achieving time $O(\sqrt{N})$ cannot be simulated directly
from $f$, because a small $\lambda$-fractional query to $f$ cannot be computed at cost
proportional to $\lambda$; the algorithm must be efficiently translated into
the discrete-query framework to be implementable.
But if we substitute the parameters into the simulation in~\cite{CleveG+2009}, we obtain a gate
cost of order $N \polylog N$ (losing the square-root speedup) and consume
order $\sqrt{N} \polylog N$ qubits of space.
The simulation in~\cite{LeeM+2011} does not appear to yield any bounds less than $O(N)$ on the gate cost.
However, our present simulation results in $N^{1/2+ o(1)}$ gates and $O(\polylog N)$ space (using the fact that the driving Hamiltonian in~\cite{FarhiG+2007} can be
implemented with $N^{o(1)}$ gates).
We remark that, for this particular example, a better simulation that is specific to AND-OR
tree evaluation (that was discovered after \cite{FarhiG+2007}) is known~\cite{Childs2007,AmbainisC+2007}.

\section{Precise Statement of Main Result}\label{sec:statement_main}

Prior to stating our main result, we give a precise definition of the
\textit{implementation cost} of a Hamiltonian acting on $l$ qubits, which is the
cost of reali{\s}ing the unitary operation corresponding to evolution under the
Hamiltonian from a start time to a finish time.
A preliminary ideali{\s}ed definition is as a unitary operation with the following
properties.
It acts on three registers: a \textit{start time}, a \textit{finish time}
and an $l$-qubit \textit{state}.
For any start and finish times $t_s$ and $t_f$,
and any $l$-qubit state $\ket{\psi}$, the unitary operation maps
$\ket{t_s}\ket{t_f}\ket{\psi}$ to $\ket{t_s}\ket{t_f}\ket{\psi'}$, where
$\ket{\psi'}$ is the state that results when $\ket{\psi}$ evolves under $H$ from
time $t_s$ to time $t_f$.
Assuming that all three registers are finite-dimensional, this can be denoted as
a gate as in Fig.~\ref{fig:ham}.
We will not require the unitary to be implemented perfectly.
We  introduce a precision parameter $\varepsilon'$,
and permit the unitary evolution to be approximated within $\varepsilon'$.
This leads to the following definition.

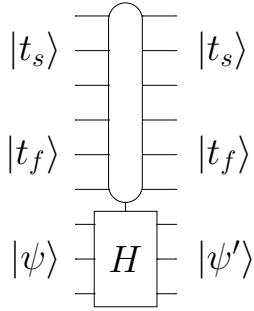
\begin{figure}
\centerline{\Qcircuit @C=0.65em @R=0.3em {
& \multimeasure{5}{\,} \qwx[6] & \qw \\
\lstick{\mbox{\Large $\ket{t_s}$}} & \ghost{\,} & \qw & \rstick{\mbox{\Large $\!\!\ket{t_s}$}}\\
& \ghost{\,} & \qw \\
& \ghost{\,} & \qw \\
\lstick{\mbox{\Large $\ket{t_f}$}} & \ghost{\,} & \qw & \rstick{\mbox{\Large $\!\!\ket{t_f}$}}\\
& \ghost{\,} & \qw \\
& \multigate{2}{\mbox{\Large $H$}} 	& \qw		\\
\lstick{\mbox{\Large $\ket{\psi}$}} & \ghost{\mbox{\Large $H$}} & \qw & \rstick{\mbox{\Large $\!\!\ket{\psi^\prime}$}}\\
& \ghost{\mbox{\Large $H$}} & \qw \\}}\vspace*{5mm}
\caption{Controlled evolution under Hamiltonian $H$, with start time $t_s$, finish time $t_f$, and target state $\ket{\psi}$.}
\label{fig:ham}
\end{figure}

\begin{definition}
Let $H$ be a Hamiltonian acting on $l$ qubits.
Define $H$ to be {\em implementable within precision $\varepsilon'$ with $G$ gates}
if the following unitary operation can be implemented within precision $\varepsilon'$
with $G$ elementary gates.
The unitary acts on three registers: a start time and finish time, and $l$ qubits set to the initial state.
The unitary maps $\ket{t_s}\ket{t_f}\ket{\psi}$ to $\ket{t_s}\ket{t_f}\ket{\psi'}$, where
$\ket{\psi'}$ is the state that results when $\ket{\psi}$ evolves under $H$ from
time $t_s$ to time $t_f$.
By approximating within $\varepsilon'$, we mean with respect to the completely bounded norm.
\end{definition}

 We are now ready to state our main result.

\begin{theorem}
\label{th:main}
{\rm\textbf{(Main)}}
Let $H(t)$ be a driving Hamiltonian that is approximately implementable within precision
$1/T$ using $G$ gates.
Then the continuous-time query algorithm can be simulated with constant error by a
discrete-query quantum algorithm
using the following resources:
\begin{itemize}
\item $O(T \log T / \log\log T)$ queries
\item $O(T G \,\log(T)+T\log^3(\|H\|T))$ 1- and 2-qubit gates
\item $O(\log^3(\|H\|T))$ qubits of space.
\end{itemize}
\end{theorem}

In particular, when $G$ is $\mbox{polylog}(T)$, this is
$\tilde{O}(T)$ queries,
$\tilde{O}(T)$ 1- and 2-qubit gates, and
$\mbox{polylog}(T)$ qubits of space.
The norm $\|H\|$ is taken to be $\|H\| := \sup_{t\in[0,T]} \|H(t)\|$ for time-dependent $H(t)$.
Because the gate complexity scales linearly in $G$, we require the driving Hamiltonian to be simulatable efficiently in order for the simulation to be gate-efficient.
If, for example, $G$ scaled linearly in $\|H\|$, then the gate complexity would be linear in $\|H\|T$, which is similar to the complexity obtained by product formulae \cite{Berry2007}.
On the other hand, we have a lower bound of $G=\Omega(\log(\|H\|T))$ (see Section \ref{sec:m}).
As a result, we could express the gate complexity as $O(TG\log(T)+TG^3)$, and the number of qubits of space as $O(G^3)$.

The remaining sections explain our algorithm, with the proof of Theorem \ref{th:main} in Section \ref{sec:proof}.

\section{Compressed CGMSY Construction}\label{sec:sketch}

We will summari{\s}e the construction in~\cite{CleveG+2009}, and then show how to
make it more efficient by compressing the control registers. Before doing so, we state the notation used throughout this paper.

\paragraph{Notation.} We denote the set of linear operators acting on complex Euclidean space $\spa{X}$ as $\LL(\spa{X})$.
The spectral norm of operator $A$ is $\snorm{A} := \max\{\norm{A\ket{v}}_2 : \norm{\ket{v}}_2 = 1\}$.
The norm of time dependent operator $A(t)$ is given by $\|A\| = \sup_{t} \|A(t)\|$.
The completely bounded norm, or diamond norm, of superoperator $\Phi: \LL(\spa{X})\mapsto \LL(\spa{Y})$ is defined as
$\norm{\Phi}_\diamondsuit = \norm{\Phi\otimes I_{\LL(\spa{X})}}_1$,
where the superoperator trace-norm is given by $\norm{\Phi}_1=\max\set{\trnorm{\Phi(X)}\mid X\in\LL(\spa{X}), \trnorm{X}\leq 1}$.
All logarithms are taken to base $2$.
We define $[m]:=\set{1,\ldots,m}$.
The tensor product of many zero computational basis states will be represented in compact form as $\ket{0^\ell}:=\ket{0}^{\otimes \ell}$.

\subsection{Overview of the CGMSY Construction~\cite{CleveG+2009}}
\label{sec:over}
Our result is obtained by simulating the construction
in~\cite{CleveG+2009}, but by representing some of the qubits in a highly compressed form.
This compressed form was known by the authors of~\cite{CleveG+2009}, but it was
not known that all of the steps of the construction can be carried out within the compressed
form---especially the \textit{measurement of control qubits}.

The construction in~\cite{CleveG+2009} begins with a continuous-time query algorithm with total query cost $T$.
The overall Hamiltonian for the continuous-time query algorithm is a sum of the oracle Hamiltonian and the driving Hamiltonian,
so the evolution can be approximated via a Lie-Trotter decomposition.
As above, it is assumed that the driving Hamiltonian can be simulated, and the evolution under the oracle Hamiltonian for a short
time becomes a fractional-time query.

The total time $T$ is partitioned into segments corresponding to time intervals of the form
$[t_0, t_0 + 1/4]$, and with $m$ of the Lie-Trotter time intervals within each segment.
We call each length $1/4$ time interval a \emph{segment}, to distinguish them from other time intervals considered.
In each of the Lie-Trotter time intervals there is a fractional query of size $1/4m$.
Here, $m$ can be chosen as a power of two without loss of generality; we henceforth assume this is the case.
In this work we consider the simulation of each of these segments.

Within each segment, there are $m$ fractional queries which we wish to simulate.
The method in~\cite{CleveG+2009} is to then, for each fractional query, use a \emph{control} qubit that is in the state $\alpha\ket{0}+i\beta\ket{1}$.
The unitary operation for the discrete oracle, $Q$, is then implemented, controlled by the control qubit.
Given that the \emph{target} system is initially in state $\ket{\zeta}$, the state after this controlled operation is
\begin{equation}
\alpha\ket{0}\otimes\ket{\zeta}+i\beta\ket{1}\otimes Q\ket{\zeta}.
\end{equation}
Finally, a projection measurement with outcome $\alpha\ket{0}+\beta\ket{1}$ yields the state in the target system (omitting normalisation)
\begin{equation}
\alpha^2\ket{\zeta}+i\beta^2 Q\ket{\zeta}.
\end{equation}

The query Hamiltonian, $H_Q$, has values on the diagonal equal to $x_j$,
whereas the discrete query unitary $Q$ has values on the diagonal of $(-1)^{x_j}=1-2x_j$.
Therefore the Hamiltonian and unitary are related by $H_Q=(I-Q)/2$.
The $I$ only gives a global phase factor and can be ignored.
Because $Q$ is self-inverse, one obtains (omitting the phase factor)
\begin{equation}
e^{-iH_Q t} = \cos(t/2) I + i\sin(t/2)Q.
\end{equation}
With $t=1/4m$, one therefore obtains the correct operation via the above procedure if $\beta\approx 1/\sqrt{8m}$.

The number of calls to the oracle can then be reduced by, instead of considering controlled operations at each time step individually, considering them jointly within a segment.
That is, considering the state of all control qubits together, for a given basis state the position of each $1$ gives a time that $Q$ is applied.
As the only basis states with significant weighting are those with a small number of ones, we can allow a maximum number $k' \in O(\log(T)/\log\log(T))$
of applications of the oracle, with evolution under the driving Hamiltonian between them.
That is, the positions of the ones in the control qubits control the time of evolution under the driving Hamiltonian.

\begin{figure}
\centerline{\Qcircuit @C=0.65em @R=0.3em {
& \lstick{\ket{0}} & \gate{R} & \gate{P} & \multigate{6}{V_1} & \qw & \multigate{6}{V_2} & \qw
& \multigate{6}{V_3} & \qw & \multigate{6}{V_{k'}} & \qw &\gate{R} & \meter &  \rstick{b_1} \cw \\
& \lstick{\ket{0}} & \gate{R} & \gate{P} & \ghost{V_1} & \qw & \ghost{V_2} & \qw & \ghost{V_3} & \qw & \ghost{V_{k'}} & \qw &\gate{R} & \meter & \rstick{b_2} \cw \\
& \lstick{\ket{0}} & \gate{R} & \gate{P} & \ghost{V_1} & \qw & \ghost{V_2} & \qw & \ghost{V_3} & \qw & \ghost{V_{k'}} & \qw & \gate{R} & \meter &  \rstick{b_3} \cw \\
& \lstick{\ket{0}} & \gate{R} & \gate{P} & \ghost{V_1} & \qw & \ghost{V_2} & \qw & \ghost{V_3} & \qw & \ghost{V_{k'}} & \qw & \gate{R} & \meter &  \rstick{b_m} \cw \\
& & \qw & \qw & \ghost{V_1} & \multigate{2}{Q} & \ghost{V_2} & \multigate{2}{Q} & \ghost{V_3} & \multigate{2}{Q} & \ghost{V_{k'}} & \multigate{2}{Q} &\qw & \qw & \qw\\
& & \qw & \qw & \ghost{V_1} & \ghost{Q} & \ghost{V_2} & \ghost{Q} & \ghost{V_3} & \ghost{Q} & \ghost{V_{k'}} & \ghost{Q} &\qw & \qw & \qw\\
& & \qw & \qw & \ghost{V_1} & \ghost{Q} & \ghost{V_2} & \ghost{Q} & \ghost{V_3} & \ghost{Q} & \ghost{V_{k'}} & \ghost{Q} &\qw & \qw & \qw\\
}}
\caption{The construction from Ref.~\cite{CleveG+2009} to simulate a segment corresponding to a time interval of length $1/4$.}
\label{fig:const}
\end{figure}
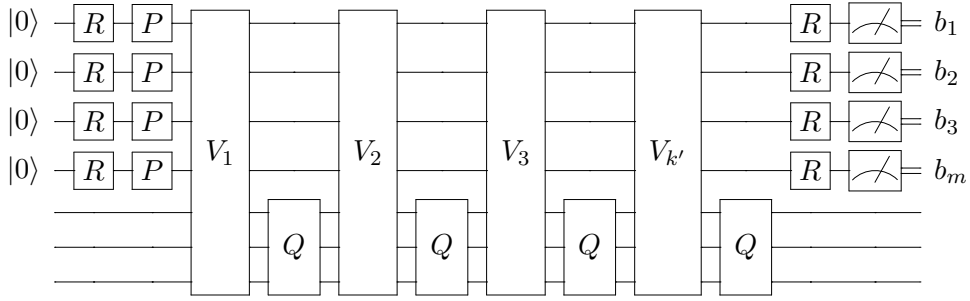

This procedure from~\cite{CleveG+2009} is represented in Fig.~\ref{fig:const}.
The operations $P$ and $R$ are designed to prepare the initial qubits and enable the final measurement, and are given by
\begin{equation}
P =
\left(
\begin{array}{lr}
1 & 0 \\
0 & i
\end{array}
\right),
\ \ \ \ \
R =
\left(
\begin{array}{lr}
\alpha & \beta \\
\beta & -\alpha
\end{array}
\right) \ \ \ \ \ \mbox{with $\beta \approx 1/\sqrt{8m}$.}
\end{equation}
The sequence of operations $PR$ acting on $\ket{0}$ prepares $\alpha\ket{0}+i\beta\ket{1}$,
and $R$ followed by a computational basis measurement of $b_j=0$ corresponds to the desired
measurement outcome $\alpha\ket{0}+\beta\ket{1}$.
The gates $V_1, \dots, V_{k'}$ are the unitaries corresponding to evolving the driving
Hamiltonian for various time intervals specified by the control qubits:
$V_1$ for the time interval from $t_0$ to the position of the first one in the control qubits;
$V_2$ for the time interval delineated by the positions of the first and second ones in the control qubits;
and so on.
The simulation is successful if $b_1 = \dots = b_m = 0$.
The probability of obtaining each $b_j=0$ is $\ge 1-1/4m$, and there are $m$ measurements,
so the probability of successful simulation is $\ge 3/4$.
The value $\beta^2 \approx 1/8m$ corresponds to a time interval of $1/4$.
This time interval is chosen to ensure that the success probability is $\ge 3/4$.

In the case that the simulation is not successful, there are errors at times corresponding to the $b_j$ that are equal to 1.
Reference \cite{CleveG+2009} shows how to correct unsuccessful instances.
Since the errors are unitary operations, it is possible to undo the step that has just been performed, then redo it.
To undo the step, one inverts the construction given in Fig.~\ref{fig:const}, but with each of the errors inverted.
This inversion will also succeed with probability $\ge 3/4$.
If this inversion does not succeed, then one attempts to undo it and then redo it, and so forth.
This procedure corresponds to a biased random walk, where a step to the right (corresponding to a success) occurs with
probability $\ge 3/4$, and a step to the left (corresponding to a failure) occurs with probability $\le 1/4$.
Overall success for this random walk is obtained when it advances one step to the right of its initial position.

That analysis continues to hold here without modification.
The only subtlety is that we also need to account for the number of gates needed to perform the gates $V_j$.
Each gate may need to be divided into a number of parts corresponding to the number of errors (ones) found.
It is shown in Ref.~\cite{CleveG+2009} that the average number of ones is $O(1)$, so the \emph{average} number of oracle queries is at most multiplied by a constant factor.
Moreover, if the total number of oracle queries permitted is bounded by 
$O(1/\varepsilon_{\rm tot})$ times the average value, then by the Markov bound, the probability is at least $1-O(\varepsilon_{\rm tot})$ that the overall correction procedure terminates within this bound \cite{CleveG+2009}.
Failure to terminate within the bound can be included in the $\varepsilon_{\rm tot}$ allowable error.
For the main result in Theorem \ref{th:main}, constant error is considered, so this does not alter the result.

When analysing the complexity due to correcting unsuccessful instances,
another factor that needs to be considered is the additional complexity due to correcting the individual errors.
The average of this complexity was denoted $C_0$ in Ref.~\cite{CleveG+2009}, but an upper bound was not considered.
As before, an upper bound equal to $O(1/\varepsilon_{\rm tot})$ times the average value will not be exceeded with probability $1-O(\varepsilon_{\rm tot})$.
Again this does not affect the result in Theorem \ref{th:main} as constant error is considered.
As a result of these considerations, when taking into account the corrections, the number of oracle queries and the number of additional 1- and 2-qubit gates are at most multiplied by a constant factor.
This means that the correction operations do not alter the scaling, and we do not need to consider them further.

The feature of the analysis in~\cite{CleveG+2009} that is most crucial for this work is that
the state of the control registers
$R^{\otimes m}\ket{0^m} = (\alpha\ket{0} + \beta\ket{1})^{\otimes m}$
is highly ``compressible" in that most of its amplitude is concentrated on basis states with
low Hamming weight.
A natural succinct representation of this state is in terms of the positions of the
ones in binary.
We first define such a succinct form precisely (Section~\ref{ssec:succinct}).
We then show how the above circuit can be simulated with the control qubits in their
succinct form in these three stages: the initial stage (Section~\ref{sec:initialization}), which is the construction of the
state $R^{\otimes m}\ket{0^m}$; the intermediate stage (Sections~\ref{seccleanup} and \ref{secphase}), where $P^{\otimes m}$ is
applied to the control qubits and then the queries and
driving operations occur; and the final stage (Section~\ref{ssec:msmt}), which is where the control qubits are
measured with respect to the basis $\{R^{\otimes m}\ket{x} : x \in \{0,1\}^m\}.$

\subsection{Succinct Representation of Control Qubits}\label{ssec:succinct}

We now propose a succinct encoding scheme which accurately reproduces low Hamming weight basis states.
Specifically, consider the set of all $m$-bit strings whose Hamming weight is at most $k+1$, where
$k$ is much smaller than $m$.
The size of this set is bounded above by $(m+1)^{k+1}$.
Our encoding scheme utilizes a set of size $(m+1)^{k+1}$ strings to accurately represent this space as follows.
We use the notation $\abs{x}$ to denote the Hamming weight of $x\in\set{0,1}^*$.
The value of $k$ is chosen to ensure that the error due to omitting high Hamming weight components is no more than $\varepsilon$, and therefore can be taken as
\begin{equation}
\label{eq:kchoice}
k = \Theta\left( \frac{\log(1/\varepsilon)}{\log\log(1/\varepsilon)} \right).
\end{equation}

We also use a slightly smaller value $k'$ to ensure that the error is no more than $\varepsilon'$; the relation between these primed variables is identical.
The Hamming weight cutoff $k$ is used to limit errors that occur repeatedly in the compressed measurement protocol.
In contrast, Hamming weight cutoff $k'$ is used to limit errors that only occur once.
In particular, we limit the total number of controlled oracle calls to $k'$, because the error due to limiting the Hamming weight there only occurs once.
We also limit the number of ones that are measured to $k'$.
The Hamming weight cutoff $k$ is used in our compressed encoding, as the error due to this cutoff will contribute multiple times.

\begin{definition}\label{def:encodingscheme}
Define the encoding scheme $C^{k}_{m}$ on $\ket{x}$ for $x \in \{0,1\}^m$, $|x| \le k$
as follows.
For $x = 0^{s_1}10^{s_2}10^{s_3}\dots 0^{s_h}10^{t}$, where $h:=|x|$, $h \le k+1$ and
$t = m - s_1 - \cdots - s_h - h$,
\begin{equation}
C^{k}_{m}\ket{x} = \ket{s_1,s_2,\dots,s_h,\underbrace{m,\dots,m}_{k+1-h}},\label{eq:Cencode}
\end{equation}
where $C^{k}_{m}\ket{x}\in(\complex^{m+1})^{\otimes k+1}$.
For $h>k+1$, $C^{k}_{m}\ket{x}$ encodes the positions of the first $k+1$ ones.
\end{definition}

\subsection{Initialization of Control Qubits in Alternative Encoding}\label{sec:initialization}

We now show how to simulate the preparation of the state after operation $R^{\otimes m}$
(but before $P^{\otimes m}$, which is deferred to Section~\ref{secphase}) in succinct form using the encoding of Definition~\ref{def:encodingscheme}.
To begin, in the original circuit there are $m$ control qubits, whose state is initiali{\s}ed to
$(\alpha\ket{0}+\beta\ket{1})^{\otimes m}$,
where $\beta \in \Theta(1/\sqrt{m})$.
The amplitudes of terms in this superposition decrease factorially with Hamming weight, and in particular, one can write
\begin{align}
(\alpha\ket{0}+\beta\ket{1})^{\otimes m} & = \sum_{x \in \{0,1\}^m} \alpha^{m - |x|}\beta^{|x|} \ket{x} \nonumber \\
& = \sum_{{x \in \{0,1\}^m} \atop {|x| \le k}} \alpha^{m - |x|}\beta^{|x|} \ket{x}
+ \sum_{{x \in \{0,1\}^m} \atop {|x| > k}} \alpha^{m - |x|}\beta^{|x|} \ket{x} \nonumber \\
& = \sum_{{x \in \{0,1\}^m} \atop {|x| \le k}} \alpha^{m - |x|}\beta^{|x|} \ket{x}
+ \mu\ket{\nu}, \label{eq:form}
\end{align}
where, on the last line, $\ket{\nu}$ is orthogonal to
every basis state in the sum that precedes it, and $\mu^2 \in 1/2^{O(k)}k!$.
Using $k$ as in Eq.\ \eqref{eq:kchoice}, $\mu^2\in O(\varepsilon)$.

Using our encoding scheme $C^{k}_{m}$, we aim to prepare a state of the form
\begin{equation}\label{eq:succinct}
\sum_{{x \in \{0,1\}^m} \atop {|x| \le k}} \alpha^{m - |x|}\beta^{|x|} C^{k}_{m}\ket{x}
+ \mu\ket{\nu^\prime},
\end{equation}
for some $\ket{\nu^\prime} \in (\complex^{m+1})^{\otimes k+1}$ orthogonal
to all the kets arising in the sum.
This state has the correct amplitudes for all encodings of strings with Hamming weights up to $k$.
Because we choose $k$ such that $\mu^2\in O(\varepsilon)$, the error in the encoding is $O(\varepsilon)$.

We now show how to construct an approximation (within distance $\varepsilon$) of the state in Eq.~(\ref{eq:succinct}) using $\poly(k,\log m)$ gates.
Note that, to accomplish this, we must avoid any approach based on first constructing
the expanded state in Eq.~(\ref{eq:form}) then applying $C^{k}_{m}$,
since this would immediately entail order $m$ gates.
Our efficient approach is to first prepare a state similar to Eq.~(\ref{eq:form}) using a slightly different encoding scheme than $C^k_m$, denoted $B^k_q$.
We then postprocess the state so that the encoding is changed from $B^k_q$ to $C^k_m$ [i.e.\ Eq.~(\ref{eq:succinct})].

We now introduce the encoding $B^k_q$ by explicit construction. Specifically, it is based on the \textit{exponential superposition} state
\begin{equation}\label{eq:exponential}
\ket{\phi_{q}} :=
\sum_{s=0}^{q-1}\beta\alpha^s\ket{s} + \alpha^q\ket{q}.
\end{equation}
The state $\ket{\phi_{q}}$ is very simple to prepare when $q = 2^r$, as follows.
Define the unitary matrix
\begin{equation}\label{eq:M}
M(\gamma) :=
\frac{1}{\sqrt{1+\gamma^2}}\left(
\begin{array}{lr}
1 & -\gamma \\
\gamma & 1
\end{array}
\right).
\end{equation}
Note that
\begin{align}
&M(\alpha^{2^{r-1}})\otimes \cdots \otimes M(\alpha^2)\otimes M(\alpha)
\ket{0^{ r}} \nonumber \\
& \quad = \frac{\beta}{\sqrt{1-\alpha^{2q}}}
\left(\ket{0\dots 00}
+ \alpha\ket{0\dots 01}
+ \alpha^2\ket{0\dots 10}
+ \cdots
+ \alpha^{q-1}\ket{1 \dots 11}\right) \ \ \ \ \nonumber \\
& \quad = \frac{1}{\sqrt{1-\alpha^{2q}}}\sum_{s=0}^{q-1} \beta\alpha^s\ket{s}.
\end{align}
Therefore, a circuit that maps $\ket{0^{r+1}}$ to $\ket{\phi_q}$ can be
obtained by first applying a one-qubit gate on the first qubit to put it in state
$\sqrt{1-\alpha^{2q}}\ket{0} + \alpha^q\ket{1}$, and then applying a sequence of controlled-$M(\alpha^{2^j})$ gates (each controlled by the first qubit being in state $\ket{0}$) to create the state
\begin{align}
& \beta\ket{0}\left(\ket{0\dots 00}
+ \alpha\ket{0\dots 01}
+ \alpha^2\ket{0\dots 10}
+ \cdots
+ \alpha^{q-1}\ket{1 \dots 11}\right) + \alpha^q\ket{1}\ket{0\cdots 00}
\nonumber \\
& \quad = \beta\left(\ket{00\dots 00}
+ \alpha\ket{00\dots 01}
+ \alpha^2\ket{00\dots 10}
+ \cdots
+ \alpha^{q-1}\ket{01 \dots 11}\right) + \alpha^q\ket{10\cdots 00}
\nonumber \\
& \quad = \ket{\phi_q}.
\end{align}

The reason why state $\ket{\phi_q}$ is useful is because, for $q \ge m$, $\ket{\phi_q}^{\otimes k+1}$
yields a state similar to Eq.~\eqref{eq:succinct}.
The encoding, which we will call $B^{k}_{q}\ket{x}$, is slightly different than $C^{k}_{m}\ket{x}$,
but can be efficiently translated into $C^{k}_{m}\ket{x}$ with some ``clean-up"
operations. Specifically, the encoding is as defined below.
\begin{definition}\label{def:encodeB}
Define the encoding scheme $B^{k}_{q}$ on $\ket{x}$ for $x \in \{0,1\}^m$, $|x| \le k$
as follows.
For $x = 0^{s_1}10^{s_2}10^{s_3}\dots 0^{s_h}10^{t}$, where $h:=|x|$, $h \le k$ and
$t = m - s_1 - \cdots - s_h - h$,
\begin{equation}
B^{k}_{q}\ket{x}
= \ket{s_1,\dots,s_h}\left(
\sum_{j=0}^{q-t-1}\alpha^{j}\beta\ket{j+t}
+ \alpha^{q-t}\ket{q}\right)\ket{\phi_q}^{\otimes k-h},\label{eq:Bencode}
\end{equation}
where $B^{k}_{q}\ket{x}\in(\complex^{q+1})^{\otimes k+1}$.
\end{definition}
The state is then given as in the following theorem.
\begin{theorem}\label{thm:encodeB}
For $q \ge m$, the states $B^{k}_{q}\ket{x}$ for $|x| \le k$ are orthonormal
and
\begin{equation}\label{eq:offset}
\ket{\phi_q}^{\otimes k+1}
= \sum_{{x \in \{0,1\}^m} \atop {|x| \le k}} \alpha^{m - |x|}\beta^{|x|} B^{k}_{q}\ket{x}
+ \mu\ket{\nu'},
\end{equation}
for some
$\ket{\nu'}$ orthogonal to all $B^{k}_{q}\ket{x}$ for $|x| \le k$.
\end{theorem}
\begin{proof}
The state $\ket{\phi_q}^{\otimes k+1}$ is a superposition of (computational)
basis states of the form
$\ket{s_1,\dots,s_{k+1}}$, where $s_1,\dots,s_{k+1} \in \{0,1,\dots,q\}$.
Intuitively, it is useful to think of each such basis state as an encoding
of a binary string $0^{s_1}10^{s_2}1 \cdots 0^{s_{k+1}}1$
(whose Hamming weight is $k+1$ and length is $s_1+\cdots+s_{k+1} + k+1$).
We will show that these basis states can be naturally partitioned into equivalence classes:
one for each prefix $x \in \{0,1\}^m$ with $|x| \le k$, and one for all
the remaining basis states.

Let $x \in \{0,1\}^m$ with $h=|x| \le k$ be of the form
$x = 0^{s_1}10^{s_2}10^{s_3}\dots 0^{s_h}10^{t}$.
Consider the set $P_x$ that consists of all $\ket{s_{1}^{\prime},s_{2}^{\prime},\dots,s_{k+1}^{\prime}}$
that are encodings of strings whose $m$-bit prefix is $x$.
The set $P_x$ consists of all $\ket{s_{1}^{\prime},s_{2}^{\prime},\dots,s_{k+1}^{\prime}}$ such
that
$(s_{1}^{\prime},s_{2}^{\prime},\dots,s_{h}^{\prime}) = (s_1,s_2,\dots,s_h)$,
$s_{h+1}^{\prime} \in \{t,\dots,q\}$, and $s_{h+2}^{\prime},\dots,s_{k+1}^{\prime} \in \{0,\dots,q\}$.
It follows that the sum of all the terms in the superposition
\begin{equation}\label{eq:full}
\ket{\phi_q}^{\otimes k+1} = \sum_{s_{1}^{\prime}=0}^q \sum_{s_{2}^{\prime}=0}^q
\dots \sum_{s_{k+1}^{\prime}=0}^q
\alpha^{s_{1}^{\prime} + s_{1}^{\prime} + \dots + s_{k+1}^{\prime}}\beta^{|\{\ell|s_{\ell}^{\prime}<q\}|}
\ket{s_{1}^{\prime},s_{2}^{\prime},\dots,s_{k+1}^{\prime}}
\end{equation}
that correspond to elements of $P_x$ is
\begin{align}
& \alpha^{s_1}\beta \cdots \alpha^{s_h}\beta
\ket{s_1,\dots,s_h}
\left(
\sum_{j=t}^{q-1}\alpha^{j}\beta\ket{j}+\alpha^q\ket{q}\right)\ket{\phi_q}^{\otimes k-h} \nonumber \\
& \quad = \alpha^{s_1}\beta \cdots \alpha^{s_h}\beta\alpha^{t}
\ket{s_1,\dots,s_h}
\left(
\sum_{j=0}^{q-t-1}\alpha^{j}\beta\ket{j+t}
+\alpha^{q-t}\ket{q}\right)\ket{\phi_q}^{\otimes k-h} \nonumber \\
& \quad = \alpha^{m-|x|}\beta^{|x|} B^{k}_{q}\ket{x},
\end{align}
which is the appropriate weighting for $B^{k}_{q}\ket{x}$ in the sum in Eq.~(\ref{eq:offset}).

Thus, the basis states in the superposition in Eq.~(\ref{eq:full}) corresponding to encodings of strings $x\in\set{0,1}^m$ of Hamming weight at most $k$ can be grouped into equivalence classes $P_x$. What about the remaining terms in $\ket{\phi_q}^{\otimes k+1}$ which do not fall in any $P_x$? These are the $\ket{s_1,\dots,s_{k+1}}$ where $s_1 + \cdots + s_{k+1} + k+1 \le m$.
Therefore, we can set
\begin{equation}
\mu\ket{\nu'} =
\sum_{s_1+\cdots+s_{k+1}+k+1 \le m}
\alpha^{s_1+\cdots+s_{k+1}}\beta^{k+1}
\ket{s_1,\dots,s_{k+1}},
\end{equation}
where $\mu\in\reals$ is chosen so that $\ket{\nu'}$ is normali{\s}ed.
All the $B^{k}_{q}\ket{x}$ and $\ket{\nu'}$ are mutually orthogonal since
they are constructed from a partition of the basis states.
\end{proof}


\subsection{Converting from the $B$ Encoding to the $C$ Encoding}\label{seccleanup}

We have thus far shown how to prepare states in the encoding $B^k_q$.
As mentioned above, we can now convert from the encoding $B^{k}_{q}$ to our desired encoding $C^{k}_{m}$. This is achieved by ``cleaning up"
the registers that follow register $h=|x|$ in $B^{k}_{q}\ket{x}$ [compare Eq.~(\ref{eq:Cencode}) with Eq.~\eqref{eq:Bencode}].
The difference is that, instead of these registers being in the state $\ket{m}$, they are in the state $\ket{\phi_q}$ (for registers $h+2$ to $k+1$).
Register $h+1$ is in a state that is similar to $\ket{\phi_{q-t}}$, except that the basis states are shifted by $t$.
Therefore, we need a way of converting these registers to the state $\ket{m}$.
However, this conversion depends on both $h$ and $t$, so we first need these quantities.

We will first give a simplified explanation, then expand on the technical details.
To determine $h$ and $t$, we compute the prefix sums
\begin{equation}\label{eq:prefix}
\ket{s_1}\ket{s_2} \cdots \ket{s_{k+1}} \mapsto
\ket{s_1+1}\ket{s_1+s_2+2} \cdots \ket{s_1 + s_2 + \cdots + s_{k+1} + k+1}.
\end{equation}
This gives the \emph{absolute} positions of the ones.
The value of $h$ can be determined by finding the first register with a value larger than $m$ (which would give a position for a one past the end of the string).

Now we can identify register $h+1$.
For this register, we wish to subtract $t$, so that the state of this register [as in Eq.~\eqref{eq:Bencode}] becomes $\ket{\phi_{q-t}}$.
At this stage we have computed the prefix sums, and subtracting $m+1$ from this modified register gives the same result as subtracting $t$ from the unmodified register.
That is, we do not need to explicitly compute $t$ to subtract it, because it is obtained implicitly in the prefix sum.
For all the other registers we then undo the prefix sums.

At this stage we have $h$ in an ancilla, and we have subtracted $t$ from register $h+1$.
Now we can undo the procedure to prepare $\ket{\phi_q}$ in registers $h+1$ to $k+1$.
Register $h+1$ is actually in state $\ket{\phi_{q-t}}$ rather than $\ket{\phi_q}$, but it is a good approximation of state $\ket{\phi_q}$.
Therefore the inverse preparation yields states $\ket{0}$ in registers $h+1$ to $k+1$, with this being approximate for register $h+1$.
It is trivial to convert $\ket{0}$ to $\ket{m}$, then uncompute the value of $h$ in the ancilla register.
This then completes the conversion of the encoding.

In summary the overall procedure is as follows.
\begin{enumerate}
\item
Compute the prefix sums.
\item
Compute $h=|x|$ in an ancilla register.
\item
Uncompute the prefix sums for registers other than $h+1$, and subtract $m+1$ from register $h+1$.
\item
Invert the procedure to prepare $\ket{\phi_q}$ from $\ket{0}$ on registers $h+1$ to $k+1$, and swap register $h+1$ with the error flag register.
\item
Flip one qubit on registers $h+1$ to $k+1$ to change $\ket{0}$ to $\ket{m}$.
\item
Uncompute $h$ in the ancilla register.
\end{enumerate}

Next we explain the technical details, including the error flag register.
When computing the prefix sums, we can first consider the case of low-Hamming weight strings with $h\le k$.
For the first $h$ registers the result is at most $m$, whereas for register $h+1$, the result is (coherently) more than $m$.
To prevent the value in register $h+1$ wrapping around modulo $m$, we instead expand the registers to dimension $m+q+2$, and perform the computations modulo $m+q+2$.
Because the value in register $h+1$ is no more than that in $h$ (which is $\le m$) plus $q+1$, the value is $\le m+q+1$, and does not wrap around modulo $m+q+2$.
The values in registers $h+2$ to $k+1$ may wrap around, but this does not affect the calculation.
This covers steps 1 and 2 above.

Next, considering step 3, the value in register $h+1$ will be
\begin{equation}
\label{eq:regh1}
s_1+\ldots+s_{h}+s_{h+1}+h+1=m-t+s_{h+1}+1.
\end{equation}
We aim to obtain $s_{h+1}-t$ in this register.
If we had computed the value of $t$, we could uncompute the prefix sums, then subtract $t$.
However, it is obvious from Eq.~\eqref{eq:regh1} that we can just subtract $m+1$ instead.
Note that this is the first register that is larger than $m$, so subtracting $m+1$ does not result in a negative number.
We also need to uncompute the prefix sums for the other registers.
This can be achieved by working backwards from register $k+1$ to $h+2$ uncomputing prefix sums, subtracting $m+1$ from register $h+1$, then uncomputing prefix sums from register $h$ back to $1$.

Next we consider the inverse preparation in step 4.
At this stage, we have subtracted $t$ from register $h+1$ yielding the exponential state
\begin{equation}
\ket{\phi_{q-t}} = \sum_{s=0}^{q-t-1} \beta\alpha^s\ket{s}
+ \alpha^{q-t}\ket{q-t}.
\end{equation}
By choosing $q$ to be sufficiently large, $\ket{\phi_{q-t}}$ is close to $\ket{\phi_{q}}$, and inverting the procedure for preparing $\ket{\phi_{q}}$ yields an accurate approximation of $\ket{0^{r+1}}$.
To be more precise, note that $\bracket{\phi_{q-t}}{\phi_q} = 1 - (1-\beta)\alpha^{2(q-t)}$.
Therefore, we have $\bracket{\phi_{q-t}}{\phi_q} \ge 1 - \varepsilon$ if
$q \ge m + (1/\beta^2)\log(1/\varepsilon)$.
To achieve this, $\ket{\phi_{q}}$ need only consist of
$\log(m + 1/\beta^2) + \log\log(1/\varepsilon)+ O(1)$ qubits.
In particular, in our context where $\beta = \Theta(1/\sqrt{m})$,
the number of
qubits is $\log m + \log\log(1/\varepsilon) + O(1)$, so the precision scales
double exponentially with the number of additional qubits beyond $\log m$.

This approximate step could alternatively be performed using the state preparation procedure of Grover and Rudolph \cite{grovrud}.
Another alternative is to use amplitude amplification to ensure that the register is set to zero correctly.
These alternatives would also not be exact, because they would require the coherent calculation of trigonometric functions.

It is convenient for the analysis to swap register $h+1$ with an ``error flag'' register that has been prepared in the $\ket{0}$ state.
Then, if this register is measured as not zero, it flags that the clean-up operation has not occurred properly.
On the other hand, register $h+1$ \emph{is} exactly $\ket{0}$.

We also need to take account of the action of the conversion procedure on the state $\ket{\nu'}$.
This state is a superposition of basis states $\ket{s_1,\ldots,s_{k+1}}$, where $s_1+\ldots+s_{k+1}+k+1\le m$.
This means that, when we compute the prefix sums, the last register will \emph{not} be $>m$.
In this case, we can set $h=k+1$, and then make no changes to the other registers in steps 3 to 5 for this value of $h$.
This means that $\ket{\nu'}$ is unchanged.
The exact form of this state is unimportant, because it corresponds to an error.
However, $\ket{\nu'}$ is a superposition of strings of Hamming weight $h+1$ encoded using $C_m^k$, and remains so under the conversion.

In summary, the overall preparation procedure is to prepare the state $\ket{\phi_q}^{\otimes k+1}$, then perform the clean-up procedure consisting of steps 1 to 6 above.
By choosing $\log q\in \Theta(\log m + \log\log (1/\varepsilon))$ (for $q$ a power of two), this then yields the state \eqref{eq:succinct}
within distance $O(\varepsilon)$.
Our circuit has size
\begin{equation}
O\left(k\left[\log m +
\log\log(1/\varepsilon)\right]\right).
\end{equation}
The final state has no values in its registers larger than $m$, so it can be stored in registers of dimension $m+1$, though higher dimensions are required in intermediate steps.

To prepare the state, we have started with all qubits of registers in the state $\ket{0}$.
It is convenient to start with these registers in the state $\ket{m}$, flip one qubit in each register to give $\ket{0}$, then perform the preparation procedure as described above.
Then we are mapping the state $C^{k}_{m}\ket{0^{m}}$ (which is the state $\ket{m}^{\otimes k+1})$ to the succinct representation of $(\alpha\ket{0}+\beta\ket{1})^{\otimes m}$ as defined in  Eq.~(\ref{eq:succinct}).

\subsection{Phase Gates, Queries and Driving Operations}
\label{secphase}
Applying the phase gates, $P^{\otimes m}$, to the control qubits in their succinct
representation is straightforward because $P^{\otimes m}\ket{x} = i^{|x|}\ket{x}$.
We need only compute $|x|$ in an ancilla register, apply $\ket{s} \mapsto i^s \ket{s}$,
and then uncompute $|x|$ in the ancilla.

To apply the driving operations, we note that our definition of driving Hamiltonian
implementation fits perfectly in this context, once we compute the prefix sums to give the positions of the ones, as in Eq.~(\ref{eq:prefix}).
In the compressed representation, $V_1$ is the implementation of the driving Hamiltonian
with $t_s$ hardwired to 0 and $t_f$ controlled by the first register.
$V_2$ is the implementation with $t_s$ controlled by the first
register and $t_f$ controlled by the second register, and so on.
At the end, the prefix sums can be uncomputed.

\subsection{The Value of $m$ Needed}
\label{sec:m}
In the CGMSY construction the number of fractional queries $m$ comes from breaking up the evolution under the oracle and the driving Hamiltonian via a product formula.
To obtain error bound by $\varepsilon_{\rm tot}$ with evolution over time $T$ and driving Hamiltonian with norm $\|H\|$, the number of time intervals needed in a Lie-Trotter-Suzuki product formula for constant Hamiltonian $H$ is $O(\|H\|T(\|H\|T/\varepsilon_{\rm tot})^{\delta})$ \cite{Berry2007}.
For the CGMSY construction the intervals need to be of equal size, which restricts $\delta$ to $1/2$.

For time-dependent Hamiltonians, the complexity of Lie-Trotter-Suzuki product formulae will depend on the magnitude of the derivatives of $H$ when one is sampling the Hamiltonians at different times \cite{Wiebe2010}.
The situation we have here is somewhat different, because we assume that the evolution under the time-dependent driving Hamiltonian can be implemented.
In this case, the error does not depend on the time derivative, and the error for a short time interval $\delta t$ can be bounded as $\|H\|\delta t^2$ 
(this is easily derived from Eq.~(2.3) of Ref.~\cite{Huy1990}).
Hence the number of intervals to limit the overall error to $O(\varepsilon_{\rm tot})$ need be no greater than $O(\|H\|T^2/\varepsilon_{\rm tot})$.
The number of intervals in one CGMSY segment of length $O(1)$ is therefore $m=O(\|H\|T/\varepsilon_{\rm tot})$.

Another question is the precision that the time needs to be specified to in order to limit the overall error to $\varepsilon_{\rm tot}$.
It is easily shown that the error in the time needs to be bound as $O(\varepsilon'/\|H\|)$ in order to limit the error in a single operation to $\varepsilon'$.
If the time is being specified on the interval $[0,T]$, then the number of bits needed for the time is $\lceil \log(\|H\|T/\varepsilon')\rceil$.
Because there are $O(1)$ controlled Hamiltonian evolutions in each CGMSY segment, we need $\varepsilon'=O(\varepsilon_{\rm tot}/T)$.
This gives the number of bits for the time as $\log(\|H\|T^2/\varepsilon_{\rm tot})+O(1)$
(where the constant $O(1)$ is because $\varepsilon'$ may have a constant of proportionality with $\varepsilon_{\rm tot}/T$).

This result is consistent with the value of $m$ used, because $\log(\|H\|T^2/\varepsilon_{\rm tot})+O(1)$ bits are needed to specify an integer from 0 to $O(mT)$.
In the CGMSY construction, a superposition over the $m$ time intervals is used, so the number of qubits needed is $\lceil\log m\rceil$.
The number of the CGMSY segment also needs to be stored, but that can be stored in $O(\log T)$ classical bits.

One can use the number of bits for the time to place a lower bound on the complexity of implementing the driving Hamiltonian.
To obtain overall accuracy $O(\varepsilon_{\rm tot})$, the driving Hamiltonian needs accuracy of $O(\varepsilon_{\rm tot}/\|H\|T)$ in the time.
There are $\Theta(\|H\|T^2/\varepsilon_{\rm tot})$ starting and finishing times, so by a counting argument
the gate complexity is $\Omega(\log(\|H\|T/\varepsilon_{\rm tot}))$.
If the driving Hamiltonian is constant, then it is only the length of the time which is important, and that is limited to $O(1)$.
The number of times is then $\Theta(\|H\|T/\varepsilon_{\rm tot})$, but the lower bound on the complexity is still $\Omega(\log(\|H\|T/\varepsilon_{\rm tot}))$.
For constant error we therefore have $G=\Omega(\log(\|H\|T))$, as used in Section~\ref{sec:statement_main}.

\section{Measurement of the Control Qubits}
\label{ssec:msmt}

What remains is to perform the final measurement.
This should logically correspond to what happens if the state is decoded from its
succinct representation to $m$ qubits
and then, for each qubit, an $R$ gate is applied and it is measured in the computational basis.
Of course, this cannot be literally implemented this way, because it would increase
the gate and space usage to at least $m$;
our task is to \textit{logically} perform this while remaining in the succinct representation.

Recall now that in Section~\ref{sec:initialization}, we constructed a procedure that approximately prepares
$R^{\otimes m}\ket{0^{m}}$ in succinct form [see Eq.~(\ref{eq:succinct})].
We define $U_m$ to be the ideal unitary that would exactly prepare the state \eqref{eq:succinct}.
The action of the ideal state preparation procedure is then
$U_m C^{k}_{m}\ket{0^{m}} \approx C^{k}_{m}R^{\otimes m}\ket{0^{m}}$.
The procedure we have described does not exactly perform this unitary, but it is within distance $O(\varepsilon)$.
Also, we do not have an exact equality, because representations of terms with Hamming weight greater than $k$ in $R^{\otimes m}\ket{0^{m}}$ are not obtained with the correct weights.
More precisely, we have
\begin{equation}
\label{eq:analyse1}
U_m C^{k}_{m}\ket{0^{m}}
= \sum_{{x \in \{0,1\}^m} \atop {|x| \le k}} \alpha^{m - |x|}\beta^{|x|} C^{k}_{m}\ket{x}
+ \mu\ket{\nu'}.
\end{equation}
This is to be compared with the uncompressed setting [Eq.~(\ref{eq:form})], in which we have
\begin{equation}
    R^{\otimes m}\ket{0^m}= \sum_{{x \in \{0,1\}^m} \atop {|x| \le k}} \alpha^{m - |x|}\beta^{|x|} \ket{x}+ \mu\ket{\nu}.\label{eq:analyse2}
\end{equation}
In terms of the \textit{logical} data, $U_m$ and $R^{\otimes m}$ produce almost
the same state when applied to $\ket{0^{m}}$.

Returning to the issue of measurement, in the \emph{uncompressed} basis we would like to perform $R^{\otimes m}$, then perform a computational basis measurement.
In the particular case that the computational basis measurement yielded all zeros, the measurement operator is $\ket{0^m}\bra{0^m}R^{\otimes m}$.
Because we are performing all operations in the compressed basis, this measurement operator can be represented by $C_m^k \ket{0^m}\bra{0^m}R^{\otimes m}(C_m^k)^\dagger$.
Because $R$ is self-inverse, this is approximately the same as $C^{k}_{m}\ket{0^m}\bra{0^{m}}(C^{k}_{m})^\dagger U_m^\dagger$.
That is, to achieve this measurement result we first invert the preparation procedure described by $U_m$.
Then, because $C^{k}_{m}\ket{0^{m}}=\ket{m}^{\otimes k+1}$ is a computational basis state, we can achieve the desired result by performing a computational basis measurement.

Ideally, this is what we want, but we also need to be able to find the positions of the ones in
the case that the all-zero string is not obtained.
At first glance, one might imagine that applying $U_m^\dagger$ in place of $R^{\otimes m}$
would yield a succinct representation of the final outcome state, so measuring in the computational basis would provide the correct result.
Unfortunately, this does not accurately simulate the final measurement except in the case where the all-zero string is obtained.
The problem is that $U_m$ and $R^{\otimes m}$ are \textit{only} in close agreement when
applied to the logical state $\ket{0^{m}}$.
For any other logical state $\ket{x}$ (for non-zero $x \in \{0,1\}^m$), applying $U_m$ and
$R^{\otimes m}$ need not yield states in any close agreement.

Our first observation towards overcoming this problem is that we can at least perform an
\textit{incomplete} measurement that captures a
seemingly small part of what we are seeking: we can cause the state to either collapse to
logical $\ket{0^m}$ or to the subspace that is the orthogonal complement of this state---and
with the correct probabilities.
This is achieved by performing $U_m^\dagger$ and then the 2-outcome incomplete projective measurement
that distinguishes between the logical state $\ket{0^m}$ and its orthogonal complement
$\ket{0^m}^{\perp}$, and then applying $U_m$ to the resulting collapsed state.
Our method to complete the measurement is to apply the above procedure recursively, on
the two halves of the logical string. We now first motivate this procedure intuitively, followed by further technical details and a rigorous proof of correctness.

\subsection{Measuring in Succinct Form: Intuition}
The intuition behind our measurement strategy is given by the following simple thought experiment.
Consider the problem of measuring an $m$-qubit state $\ket{\psi}$ in the computational basis.
This can be accomplished by performing a sequence of two-outcome measurements in a variety
of ways.
One obvious approach is to measure the state of the first qubit, then the second qubit, and so on.
Each final outcome $x \in \{0,1\}^m$ will occur with exactly the same probability as with the original
complete measurement.
We now describe an alternative---and unconventional---approach for simulating the same
measurement.

First, perform the measurement distinguishing between $\ket{0^m}$ and $\ket{0^m}^{\perp}$, its
orthogonal complement.
If the state collapses to $\ket{0^m}$ we halt, outputting $0^m$.
Otherwise (when the state collapses to $\ket{0^m}^{\perp}$), apply the measurement
$\ket{0^{m/2}}$ vs.\ $\ket{0^{m/2}}^{\perp}$ to the first $m/2$ qubits.
If that part of the state collapses to $\ket{0^{m/2}}$ then output $0^{m/2}$ for the first
$m/2$ bits; otherwise recurse further.
Once this recursive measurement procedure for the first $m/2$ qubits has terminated, repeat it for
the second $m/2$ qubits.
Each final outcome $x \in \{0,1\}^m$ occurs with exactly
the same probability as with the original complete measurement.
Note that although this process may appear complicated, it terminates fast whenever the Hamming weight of
the final outcome $x$ is small: for Hamming weight up to $k'$, at most $k' \log m$ steps are performed.

Our actual scenario is different than the one described above in that the final measurement is in
the basis $\{ R^{\otimes m}\ket{x} : m \in \{0,1\}^m\}$ rather than the computational basis.
However, our logical $U_m$ and $U_m^\dagger$ permit us to approximate the
$R^{\otimes m}\ket{0^m}$ vs.\ $R^{\otimes m}\ket{0^m}^{\perp}$ measurement well.
Also, making use of the fact that the underlying operation that we are simulating
has a tensor product structure,
$R^{\otimes m}\ket{x_1x_2} = R^{\otimes m/2}\ket{x_1}R^{\otimes m/2}\ket{x_2}$ for
any $x_1, x_2 \in \{0,1\}^{m/2}$, we can emulate the recursive procedure in the above thought
experiment. We now make this rigorous.

\subsection{Measuring in Succinct Form: Details}
We now introduce Alg.~\ref{alg:measure}, which formali{\s}es the intuition behind the recursive measurement outlined above, and show that it simulates the desired measurement in succinct form. Recall that we assume without loss of generality that $m$ is a power of $2$.

Before stating Alg.~\ref{alg:measure}, we require a lemma which allows us to efficiently ``split'' the encoded version of string $x=x_1x_2$ into the concatenation of the encoded versions of $x_1$ and $x_2$.
\begin{lemma}\label{lem:split}
    Let $x=x_1x_2$ for $x\in\set{0,1}^m$ with $\abs{x}\leq k$, $x_1,x_2\in\set{0,1}^{m/2}$ and $m$ a power of $2$. Then there exists a quantum circuit with complexity $O(k\log m)$ for achieving the mapping
    \begin{equation}
        C^{k}_{m}\ket{x_1x_2}\mapsto C^{k}_{m/2}\ket{x_1}\otimes C^{k}_{m/2}\ket{x_2},
    \end{equation}
    where $C^{k}_{m}\ket{x_1x_2}$, $C^{k}_{m/2}\ket{x_1}$, $C^{k}_{m/2}\ket{x_2}\in (\complex^{m+1})^{\otimes k+1}$.
\end{lemma}
\begin{proof}
Because both $C^{k}_{m}\ket{x_1x_2}$ and $C^{k}_{m/2}\ket{x_1}\otimes C^{k}_{m/2}\ket{x_2}$ are computational basis states, the procedure that is performed is the same as would be performed classically, except that it must be performed coherently.
That is, there is a reversible classical procedure to split the encoding in the computational basis,
which immediately provides a coherent procedure for splitting the encoding.
Because there are $O(k)$ registers of size $O(\log m)$, the complexity of this procedure is $O(k\log m)$.
\end{proof}

\begin{figure}[t]
\noindent\rule{\linewidth}{0.3mm}
\begin{algorithm}\rm{$S$ = MEASURE( }$A$ , $m_1$ , $m_2$ \rm{ )}.\label{alg:measure}
    \begin{itemize}
        \item Input:\hspace{2mm}$A$ -- Registers corresponding to space $(\complex^{m+1})^{\otimes k+1}$ containing the subset $\set{m_1,\ldots, m_2}$ \\ \mbox{\hspace{20mm}}of the encoded control qubits.\\
        \mbox{\hspace{12mm}}$m_1$ -- The starting index $m_1\in[m]$ of the encoded qubits in $A$.\\
        \mbox{\hspace{12mm}}$m_2$ -- The ending index $m_2\in[m]$ of the encoded qubits in $A$.
\item Precondition: $m_2-m_1+1$ is a power of two.
        \item Output: A set of indices $S\subseteq[m]$ containing the positions where an uncompressed measurement would have found ones in the uncompressed setting.
    \end{itemize}

Perform a measurement described by the measurement operators $\measo{\comp}{0}^{m_2-m_1+1}$ and $\measo{\comp}{1}^{m_2-m_1+1}$, where $\measo{\comp}{0}^n := U_n C_n^k \ket{0^{n}}\bra{0^{n}}(C_n^k)^\dagger U_n^\dagger$
and $\measo{\comp}{1}^n:=I-\measo{\comp}{0}^n$.
Label the measurement result $d$. Then

    \begin{compactenum}
\item (Zero detected) If $d=0$: Return $S=\emptyset$.
        \item (Base case) If $d=1$ and $m_1=m_2$:
Return $S=\set{m_1}$.
        \item (Recurse) If $d=1$ and $m_2>m_1$:
Split $A$ to $A_1$ and $A_2$, containing the encoded forms of the first and second halves, respectively, of the control qubits. Then return
                \begin{equation}
                    S=\operatorname{MEASURE}(A_1,m_1,(m_1+m_2-1)/2)\cup\operatorname{MEASURE}(A_2,(m_1+m_2+1)/2,m_2).
                \end{equation}
    \end{compactenum}
\end{algorithm}
\noindent\rule{\linewidth}{0.3mm}
\end{figure}

The formal statement of the recursive measurement algorithm is given in Alg.~\ref{alg:measure}.
To perform our recursive measurement, we simply call $\operatorname{MEASURE}(A,1,m)$, where $A$ is the register containing our compressed control qubits.
Once the procedure finishes running, it will return the locations of all the ones an uncompressed measurement would have obtained when measuring the uncompressed version of $A$.
We truncate the recursive measurement procedure if $k'$ ones have been located, to limit the complexity of the procedure.

We now introduce a notation that will be used throughout the remainder of the paper in order to simplify reference to quantities in the uncompressed protocol versus the compressed protocol.
For quantities (states, operators or probabilities) in the compressed protocol, we will use a superscript or subscript ``$\comp$'', whereas we will use ``$\unco$'' for the uncompressed protocol.
To refer to quantities defined for both, we will use ``$\either$''.
We also use $n$ to refer to operations acting on a compressed sub-portion of the string of length $n$ (instead of $m$ for the full string).

To perform the measurement described by the measurement operators $\measo{\comp}{d}^n$ in Alg.~\ref{alg:measure}, we apply $U_n^\dagger$, perform the measurement that distinguishes the encoded all-zero state from all other states, then apply $U_n$.
In this form it is clear why we need to perform the operation $U_n$ after the measurement:
it means that all states orthogonal to that corresponding to measurement result $0$ are unchanged, because they are just acted upon by the identity.
The final $U_n$ operation is also included for the $0$ measurement result for simplicity, but it is not needed.
As these measurement operators are projections, they are the same as the positive operator-valued measure elements.

For simplicity we have described the measurement in terms of $U_n$ and $U_n^\dagger$, but in reality we will use operations on an expanded space that includes an error-flag ancilla.
Recall that, because $\ket{\phi_{q-t}}$ is not exactly equal to $\ket{\phi_q}$, we have a register that is not exactly reset to zero, and this is swapped into an ancilla register.
The unitary operations in this expanded space will be denoted $\widetilde U_n$ and $\widetilde U_n^\dagger$.
Then the action of $\widetilde U_n$ is
\begin{equation}
\widetilde U_n C_n^k \ket{0}\otimes \ket{0} = \sum_{{x \in \{0,1\}^n} \atop {|x| \le k+1}} \ca_x^n \left[\sqrt{1-\varepsilon_x}C_n^k  \ket{x}\otimes \ket{0} + \sqrt{\varepsilon_x}\ket{{\rm err}_x}\otimes\ket 1\right].
\end{equation}
Here the tensor product with $\ket{0}$ on the left-hand side indicates the use of ancillas that are initially in the state zero.
The amplitudes $\ca_x^n$ are the amplitudes for each $C_n^k  \ket{x}$ in the state \eqref{eq:succinct} (when $m$ is replaced with $n$).
These amplitudes include those for $|x|=k+1$ for the state $\ket{\nu'}$, which corresponds to encoded Hamming-weight $k+1$ states.
For $|x|\le k$, we have $\ca_x^n=\alpha^{n-|x|}\beta^{|x|}$.
The tensor product with $\ket{0}$ on the right-hand side indicates ancillas that will be set to zero in the case of success.
The parameter $\varepsilon_x$ is $\le \varepsilon$, and can in general depend on $x$.
The state $\ket{{\rm err}_x}$ is an error state.

For the ideal state preparation, we have
\begin{equation}
\bra{x} (C_n^k)^\dagger  U_n C_n^k \ket{0} = \ca_x^n.
\end{equation}
Using the expression for the action of $\widetilde U_n$, we find
\begin{equation}
[(\bra{x} (C_n^k)^\dagger ) \otimes \bra{0}][ \widetilde U_n C_n^k \ket{0}\otimes \ket{0}] = \ca_x^n \sqrt{1-\varepsilon_x} = \bra{x} (C_n^k)^\dagger  U_n C_n^k \ket{0} [1-O(\varepsilon)].
\end{equation}
There is no contribution from the error register, because the error flag is orthogonal to zero for that register.

To perform the measurement, we append ancillas in the zero state, and perform $\widetilde U_n^\dagger$.
Then we perform the measurement that projects onto $C_n^k \ket{0^{n}}\otimes\ket{0}$ and its orthogonal complement.
Here the tensor product with $\ket 0$ indicates the extra ancillas used by the full preparation procedure $\widetilde U_n$.
Then we perform $\widetilde U_n$.

The action of this measurement will have error $O(\varepsilon)$ from that used in the algorithm.
First, consider the resulting state for zero measurement result and initial state $C_n^k\ket{x}$.
\begin{align}
\widetilde U_n [C_n^k \ket{0^{n}}\bra{0^{n}}(C_n^k)^\dagger \otimes\ket{0}\bra{0}] \widetilde U_n^\dagger C_n^k\ket{x} \ket{0} 
&= \widetilde U_n C_n^k \ket{0^{n}}\otimes \ket{0}[\bra{x}  \bra{0}(C_n^k)^\dagger  \widetilde U_n C_n^k \ket{0^n} \ket{0}]^* \nonumber \\
&= \widetilde U_n C_n^k \ket{0^{n}}\otimes \ket{0}[\bra{x} (C_n^k)^\dagger  U_n C_n^k \ket{0^n}]^* [1-O(\varepsilon)] \nonumber \\
&= \widetilde U_n C_n^k \ket{0^{n}}\otimes \ket{0}[\bra{0^n} (C_n^k)^\dagger  U_n^\dagger C_n^k \ket{x}] [1-O(\varepsilon)] .
\end{align}
Therefore we find that the probability of this result is changed by no more than $O(\varepsilon)$.
In addition, tracing over the ancillas used, $\widetilde U_n C_n^k \ket{0^{n}}\otimes \ket{0}$ is an approximation of $U_n C_n^k \ket{0^{n}}$ with trace distance $O(\varepsilon)$.
Therefore, for the zero measurement result, the resulting state has trace distance no more than $O(\varepsilon)$ from that for the ideal measurement using $U_n$.

The resulting state for measurement result $1$ is then
\begin{equation}
C_n^k\ket{x} \ket{0} - \widetilde U_n C_n^k \ket{0^{n}}\otimes \ket{0}[\bra{0} (C_n^k)^\dagger  U_n^\dagger C_n^k \ket{x}] [1-O(\varepsilon)].
\end{equation}
This is because $\widetilde U_n$ is exactly the inverse of $\widetilde U_n^\dagger$ in the expanded space.
Trivially from the result for the zero measurement result, once we trace over the ancilla the resulting state has trace distance no more than $O(\varepsilon)$ from that for the ideal measurement.
As a result, even though we can not perform $U_n$ exactly, we can approximate the measurements within error $\varepsilon$ using the $\widetilde U_n$ and $\widetilde U_n^\dagger$ operations.

To show that the algorithm correctly simulates the desired uncompressed measurement, we consider a similar recursive measurement on the uncompressed state.
We show that, except for the imprecision due to approximating $U_n$ and omitting high Hamming weight components, the low Hamming weight portions of the states in Eqs.~\eqref{eq:analyse1} and~\eqref{eq:analyse2} evolve identically.
Moreover, this holds even if the control qubits are entangled with a target register, as is generally the case here.

In the uncompressed setting, the state of the control and target registers before the final measurement can be described as approximately
\begin{equation}\label{eq:ent1}
    \ket{\instalo{\unco}}:=\sum_{{x \in \{0,1\}^m} \atop {|x| \le k}} \co_{x,0}^{()} \ket{x}\ket{w_x},
\end{equation}
where $\ket{w_x}$ describes the state of the target register where the queries $Q$ are applied.
Note that $\ket{\instalo{\unco}}$ is unnormali{\s}ed, as we have omitted the high Hamming weight component.
Similarly, in the compressed setting, before the final measurement we approximately have the (unnormali{\s}ed) state
\begin{equation}\label{eq:ent2}
    \ket{\instalo{\comp}}:=\sum_{{x \in \{0,1\}^m} \atop {|x| \le k}} \co_{x,0}^{()} C^{k}_{m}\ket{x}\ket{w_x},
\end{equation}
where the states $\ket{w_x}$ coincide with those in the uncompressed case.
The coefficients $\co_{x,0}^{()}$ are the same in each case, and are equal to $i^{|x|}\ca_x^m$.
We use this notation for consistency with the coefficients for the intermediate states in Eqs.~\eqref{eq:ent3} and \eqref{eq:ent4} below.

We consider a measurement in the uncompressed case that is the same as in Alg.~\ref{alg:measure}.
We show that the results obtained in the two cases are close, but there are two sources of error:
(1) the error incurred due to the high Hamming weight component of the state, and (2) the error due to not implementing $U_n$ exactly.
First we discuss the \emph{error-free} case, i.e.\ where (1) we omit the high Hamming weight component, and where (2) $U_n$ is implemented exactly.
We subsequently reintroduce both sources of error and analy{\s}e their impacts.
In the error-free analysis, we show the following.

\begin{theorem}[Error-free simulation]\label{thm:errorfree}
Assume we are in the error-free setting defined above.
Then, suppose that before the final measurement, the states of the uncompressed and compressed control and target qubits are given by Eqs.~(\ref{eq:ent1}) and (\ref{eq:ent2}), respectively.
Then, Alg.~\ref{alg:measure} exactly simulates the uncompressed $R^{\otimes m}$ measurement in the following sense:
    \begin{enumerate}
        \item After running Alg.~\ref{alg:measure}, the probability of obtaining a given measurement result is the same as for the uncompressed $R^{\otimes m}$ measurement, and
        \item for a given measurement result the state of the target register in both uncompressed and compressed settings matches.
    \end{enumerate}
\end{theorem}

\begin{proof}
Measuring $R^{\otimes m}\ket{\widetilde{\psi}}$ in the computational basis can also be simulated using a recursive approach; namely, we apply $R^{\otimes m}$, followed by the incomplete measurement of $\ket{0^{m}}$ versus its orthogonal complement, then apply $R^{\otimes m}$.
This can be represented by the measurement operators $\measo{\unco}{d}^m$, with
\begin{equation}
\measo{\unco}{0}^n := R^{\otimes n}\ket{0^{n}}\bra{0^{n}}R^{\otimes n},
\end{equation}
and $\measo{\unco}{1}^n:= I-\measo{\unco}{0}^n$.
Similar to Alg.~\ref{alg:measure}, we are including the application of $R^{\otimes m}$ for both measurement results for simplicity, though it is not needed for result $0$.
If we obtain $1$ as the outcome, we recurse on the two blocks of $m/2$ qubits by applying the measurement with operators $\measo{\unco}{d}^{m/2}$, and so forth.

To prove the result, we simply need to show that at each step in the recursion the states resulting from measurement operators $\measo{\comp}{d}^n$ and $\measo{\unco}{d}^n$ are equivalent.
Let us denote the measurement result obtained at each step in the recursive measurement scheme by $d_j$.
Then, at step $\ell$, we have measurement results $d_1,\ldots,d_{\ell-1}$, and will have a state that depends on those measurement results.
Let us assume that at this step we have equivalent states for the compressed and uncompressed cases.
The base case is that for $\ell=1$, where the initial states \eqref{eq:ent1} and \eqref{eq:ent2} are equivalent.
Then the states for the two cases can be expressed as
\begin{align}\label{eq:ent3}
    \ket{\midstate{\comp}{\ell-1}}&=\sum_{|x| \le k} \co_{x,\ell-1}^{(d_1,\ldots,d_{\ell-1})} (C^{k}_{n}\otimes C^{k}_{\rm rest}) \ket{x}\ket{w_x}, \\
\label{eq:ent4}
    \ket{\midstate{\unco}{\ell-1}}&=\sum_{|x| \le k} \co_{x,\ell-1}^{(d_1,\ldots,d_{\ell-1})}  \ket{x}\ket{w_x}.
\end{align}
At this stage the encoding will be a succinct encoding on a subset of $n$ of the digits of $x$, and another encoding of the remaining digits (denoted $C^{k}_{\rm rest}$), the exact form of which is unimportant for this analysis.
The subset of $n$ of the digits of $x$ will depend on $d_1,\ldots,d_{\ell-1}$.
This dependence has not been indicated here for brevity.
We also omit $x \in \{0,1\}^m$ from the sum for brevity.

In order for the results obtained for the compressed and uncompressed cases to be equivalent, all that is required is that the amplitude weightings $\co_{x,\ell-1}^{(d_1,\ldots,d_{\ell-1})}$ in Eqs.\ \eqref{eq:ent3} and \eqref{eq:ent4} are the same.
The results are equivalent in the sense that the probability of the measurement results, as well as the state of the target system for a given measurement result, are the same.
The probability of the measurement results will be obtained from the normali{\s}ation of the state, which must be the same if the amplitudes are the same.
Similarly the resulting state in the target system will be the same if the amplitudes are the same.

We will adopt the notation that $I_{\rm rest}$ indicates the identity on the remaining registers, so the overall measurement operator is $\measo{\comp}{d}^n\otimes I_{\rm rest}$.
We will also adopt the notation that $x_n$ is the subset of $n$ digits of the string $x$, and $x_{\rm rest}$ is the remaining digits.
Then we have
\begin{equation}
    \bra{0^n} R^{\otimes n}\ket{x_n}=\bra{x_n}R^{\otimes n}\ket{0^n} =\alpha^{n-\abs{x_n}}\beta^{\abs{x_n}}=\bra{x_n}(C^k_n)^\dagger U_n C^k_n\ket{0^n} =\bra{0^n}(C^k_n)^\dagger U_n^\dagger C^k_n\ket{x_n}.
\end{equation}
For the compressed case, consider performing the measurement with operators $\measo{\comp}{d}^n$.
In the case that the measurement result is $d=0$, our compressed state becomes
    \begin{align}
&(\measo{\comp}{0}^n \otimes I_{\rm rest})\ket{\midstate{\comp}{\ell-1}}
\approx \left(U_nC^{k}_{n}\ketbra{0^n}{0^n}(C^{k}_{n})^\dagger U_n^\dagger \otimes I_{\rm rest}\right) \left[\sum_{|x| \le k} \co_{x,\ell-1}^{(d_1,\ldots,d_{\ell-1})} (C^{k}_{n}\otimes C^k_{\rm rest})\ket{x}\ket{w_x}\right]\nonumber\\
        &\qquad =\sum_{\abs{x}\leq k}\co_{x,\ell-1}^{(d_1,\ldots,d_{\ell-1})}\left(\bra{0^n}(C^{k}_{n})^\dagger U_n^\dagger C^{k}_{n}\ket{x_n}\right)U_nC^{k}_{n}\ket{0^n}C^k_{\rm rest}\ket{x_{\rm rest}}\ket{w_x}\nonumber\\
        &\qquad=\sum_{\abs{x}\leq k}\co_{x,\ell-1}^{(d_1,\ldots,d_{\ell-1})}\bra{0^n}R^{\otimes n}\ket{x_n} U_nC^{k}_{n}\ket{0^n}C^k_{\rm rest}\ket{x_{\rm rest}}\ket{w_x}\nonumber\\
        &\qquad\approx\sum_{\abs{x}\leq k}\co_{x,\ell-1}^{(d_1,\ldots,d_{\ell-1})}\bra{0^n}R^{\otimes n}\ket{x_n} \left(\sum_{\abs{y}\leq k}\ca_y^n C^{k}_{n}\ket{y}\right)C^k_{\rm rest}\ket{x_{\rm rest}}\ket{w_x}
=: \ket{\tilde\psi_{\comp,\ell}^{(d_1,\ldots,d_{\ell-1},0)}},
\label{eq:comap}
    \end{align}
where $\ca_y^n=\alpha^{n-|y|}\beta^{|y|}$, and $y$ is an $n$-digit string.
The approximate equality in the first line of Eq.~\eqref{eq:comap} is because the measurement operator $\measo{\comp}{0}^n$ cannot be obtained exactly, because the unitary $U_n$ is not performed exactly.
The approximate equality in the last line is because the high Hamming weight components have been omitted.
In the error-free setting the error in these approximations is ignored.

In comparison, in the uncompressed setting, a similar calculation yields, for $d=0$,
    \begin{align}
& (\measo{\unco}{0}^n \otimes I_{\rm rest})\ket{\midstate{\unco}{\ell-1}}=(R^{\otimes{n}}\ketbra{0^n}{0^n}R^{\otimes{n}} \otimes I_{\rm rest}) \ket{\midstate{\unco}{\ell-1}}\nonumber \\
        &\qquad \approx\sum_{\abs{x}\leq k}\co_{x,\ell-1}^{(d_1,\ldots,d_{\ell-1})}\bra{0^n}R^{\otimes n}\ket{x_n}\left(\sum_{\abs{y}\leq k}\ca_y^n \ket{y}\right)\ket{x_{\rm rest}}\ket{w_x} 
=: \ket{\psi_{\unco,\ell}^{(d_1,\ldots,d_{\ell-1},0)}}.
\label{eq:unap}
    \end{align}
The approximate equality in the last line is again due to omitting high Hamming weight components.
In the error-free setting the error in this approximation is ignored.
In the case that the measurement result is $d=1$, then the states obtained are
\begin{align}
\left( I - \measo{\comp}{0}^n\otimes I_{\rm rest}\right)\ket{\midstate{\comp}{\ell-1}} = \ket{\midstate{\comp}{\ell-1}} - \ket{\tilde\psi_{\comp,\ell}^{(d_1,\ldots,d_{\ell-1},0)}} =: \ket{\tilde\psi_{\comp,\ell}^{(d_1,\ldots,d_{\ell-1},1)}}, \nonumber \\
\left( I - \measo{\unco}{0}^n\otimes I_{\rm rest}\right)\ket{\midstate{\unco}{\ell-1}} = \ket{\midstate{\unco}{\ell-1}} - \ket{\psi_{\unco,\ell}^{(d_1,\ldots,d_{\ell-1},0)}} =: \ket{\psi_{\unco,\ell}^{(d_1,\ldots,d_{\ell-1},1)}}.
\end{align}

Above we have defined resulting states after the measurements in the uncompressed and compressed setting of $\ket{\midstate{\unco}{\ell}}$ and $\ket{\tilde\psi_{\comp,\ell}^{(d_1,\ldots,d_{\ell})}}$, respectively.
The quantity $\ket{\tilde\psi_{\comp,\ell}^{(d_1,\ldots,d_{\ell})}}$ is the state in the compressed case before the change in the compression.
To obtain the state $\ket{\midstate{\comp}{\ell}}$, the compression of the string must be changed as per Lemma~\ref{lem:split}.
This can be done without error, and does not change the amplitudes.

Omitting the high Hamming weight states, we start with states $\ket{\instalo{\either}}$, which have the same amplitudes in the compressed and uncompressed cases.
Then, by the above reasoning, if the amplitudes are the same at step $\ell-1$, they are the same at step $\ell$.
Therefore, by induction, the amplitudes must be the same after the full recursive measurement.
Therefore the same amplitudes are obtained for the compressed and uncompressed cases, so the results obtained in the compressed and uncompressed cases are equivalent.
That is, the probabilities of the measurement results and the state of the target register for a given measurement result match.
\end{proof}

Theorem~\ref{thm:errorfree} shows that if we focus solely on the low Hamming weight subspace, and if we assume we can prepare the state $C^k_n\ket{0^n}$ exactly, then our succinct recursive measurement Alg.~\ref{alg:measure} \emph{perfectly} simulates the uncompressed measurement.
We now analy{\s}e the error incurred when these two assumptions are dropped.
First we need to identify the appropriate measure of the error in the measurement.
We would like to bound the average trace distance; i.e.
\begin{equation}
\overline{D} := \sum_{\bf b} \pr{\comp} \| \row{\unco}- \row{\comp}\|_{\rm tr} ,
\end{equation}
where $\pr{\either}$ is the probability of obtaining the measurement result ${\bf b}=(b_1,\ldots,b_m)$, and $\row{\either}$ is the state for the target system.
We would also like to bound the error in the probabilities obtained.
This is because measurement results with many ones will be difficult to correct, so we need to ensure that the probabilities for those measurement results remain small.
The error in the probability distribution can be quantified by
\begin{equation}
\Delta p := \sum_{\bf b} |\pr{\unco}-\pr{\comp}|.
\end{equation}

We can bound both those errors using the quantity
\begin{equation}
\label{eq:ddef}
D_{\rm av} := \sum_{\bf b} \| \pr{\unco}\row{\unco} - \pr{\comp}\row{\comp} \|_{\rm tr}.
\end{equation}
Because the trace distance is non-increasing under channels, and we obtain $\Delta p$ by applying the completely depolari{\s}ing channel to both $\row{\unco}$ and $\row{\comp}$ in Eq.\ \eqref{eq:ddef}, we have $\Delta p \le D_{\rm av}$.
Then we have
\begin{equation}
\pr{\comp}\| \row{\unco} - \row{\comp} \|_{\rm tr} \le \| \pr{\comp}\row{\unco} - \pr{\unco}\row{\unco} \|_{\rm tr} + \| \pr{\unco}\row{\unco} - \pr{\comp}\row{\comp} \|_{\rm tr}
\le 2 \| \pr{\unco}\row{\unco} - \pr{\comp}\row{\comp} \|_{\rm tr}.
\end{equation}
Summing over ${\bf b}$ then gives $\overline{D} \le D_{\rm av}$.

\begin{theorem}[Error bounds]\label{thm:error}
The error between compressed and uncompressed schemes can be bounded as
\begin{equation}
\label{eq:erbnd}
D_{\rm av} = O(\varepsilon'+\varepsilon k'\log m).
\end{equation}
\end{theorem}
\begin{proof}
In order to bound the value of $D_{\rm av}$, we have two main sources of error.
First is that in preparing the initial state, where the high Hamming weight terms are omitted, and second is the sequence of approximations in the measurement operators in Eqs.\ \eqref{eq:comap} and \eqref{eq:unap}.
The approach to bounding the error is as follows.
In locating the position of a single one in the measurement result, there is a contribution of $O(\varepsilon)$ to the error from each of the steps as described in Eq.\ \eqref{eq:comap}.
These need to be performed $\log m$ times, and as a result the contribution to the error is $O(\varepsilon\log m)$.
If $h$ ones need to be located, the worst case is where the sequence of measurements to locate these ones is independent, so the contribution to the error is $O(h\varepsilon\log m)$.
Since the error due to locating no more than $k'$ ones will be $O(\varepsilon')$, we can take $h\le k'$, and bound the overall error by $O(\varepsilon k' \log m+\varepsilon')$.

To make this analysis rigorous, we first want to omit the high Hamming weight measurement results.
For the measurements in the uncompressed case, the probability of measurement results with Hamming weight over $k'$ is $O(\varepsilon')$.
This is because the probability of obtaining each one is no more than $2\alpha^2\beta^2$.
Because we take $\beta^2\approx 1/8m$, the probability of obtaining more than $k'$ ones
with $k' = \Theta(\log(1/\varepsilon') / \log\log(1/\varepsilon'))$ is $O(\varepsilon')$.
Recall that we place a bound $\varepsilon'$ on errors that only occur once in each time step, and use a corresponding Hamming weight cutoff $k'$, whereas we use $k$ for limiting errors that occur multiple times in the measurement process.

To bound $D_{\rm av}$, we also need to take account of the probability of high Hamming weight measurement results for the uncompressed measurement.
We can do this in the following way.
First use
\begin{equation}
\sum_{|{\bf b}|> k'} (\pr{\comp} - \pr{\unco}) = \sum_{|{\bf b}|\le k'} (\pr{\unco} - \pr{\comp}) \le \sum_{|{\bf b}|\le k'} |\pr{\unco} - \pr{\comp}| \le  \sum_{|{\bf b}|\le k'} \| \pr{\unco}\row{\unco} - \pr{\comp}\row{\comp} \|_{\rm tr}.
\end{equation}
Therefore we can bound $D_{\rm av}$ by
\begin{align}
D_{\rm av} &\le \sum_{|{\bf b}|> k'} (\pr{\comp}+ \pr{\unco}) + \sum_{|{\bf b}|\le k'} \| \pr{\unco}\row{\unco} - \pr{\comp}\row{\comp} \|_{\rm tr} \nonumber \\
&= \sum_{|{\bf b}|> k'} (\pr{\comp}- \pr{\unco}) +2\sum_{|{\bf b}|> k} \pr{\unco} +
 \sum_{|{\bf b}|\le k'} \| \pr{\unco}\row{\unco} - \pr{\comp}\row{\comp} \|_{\rm tr} \nonumber \\
&\le O(\varepsilon') + 2\sum_{|{\bf b}|\le k'} \| \pr{\unco}\row{\unco} - \pr{\comp}\row{\comp} \|_{\rm tr}.
\label{eq:erhigh}
\end{align}
This means that omitting the high Hamming weight measurement results can only change the results by a multiplying factor and an $O(\varepsilon')$ term.
For convenience we define
\begin{equation}
D'_{\rm av} :=  \sum_{|{\bf b}|\le k'} \| \pr{\unco}\pr{\unco} - \pr{\comp}\row{\comp} \|_{\rm tr}.
\end{equation}

Next we note that the distance measure can be written as a trace distance between two states, rather than the average of trace distances.
That is,
\begin{equation}
D'_{\rm av} =   \left\| \sum_{|{\bf b}| \le k'} \left(\pr{\unco}\ket{{\bf b}}\bra{{\bf b}}\otimes \row{\unco} - \pr{\comp}\ket{{\bf b}}\bra{{\bf b}}\otimes \row{\comp} \right) \right\|_{\rm tr}.
\end{equation}
The reason for this is that the complete matrix is block-diagonal, with $\pr{\unco}\pr{\unco} - \pr{\comp}\row{\comp}$ as the blocks on the diagonal.
The trace distance for the entire density matrix is just the sum of the trace distances for the blocks on the diagonal, which is the definition of $D'_{\rm av}$.

Let us denote by $\ket{\instate{\either}}$ the states obtained after preparation and controlled operations.
Then we have
\begin{equation}
\pr{\either}\row{\either} = \trc (\measb{\either} \ket{\instate{\either}}\bra{\instate{\either}}\measb{\either}^\dagger).
\end{equation}
Here $\trc$ indicates a trace over the control registers.
Then we have
\begin{equation}
D'_{\rm av} = \left\| \sum_{|{\bf b}|\le k'} \left[\ket{{\bf b}}\bra{{\bf b}}\otimes\trc \left(\measb{\unco} \ket{\instate{\unco}}\bra{\instate{\unco}}\measb{\unco}^\dagger\right) -  \ket{{\bf b}}\bra{{\bf b}}\otimes\trc \left(\measb{\comp} \ket{\instate{\comp}}\bra{\instate{\comp}}\measb{\comp}^\dagger\right)\right] \right\|_{\rm tr}.
\end{equation}

Now note that the maps defined by
\begin{equation}
\chant{\either}(\rho) := \sum_{\bf b} \ket{{\bf b}}\bra{{\bf b}}\otimes\trc (\measb{\either} \rho \measb{\either}^\dagger) ,
\end{equation}
are completely-positive trace-preserving (CPTP).
This means that trace distance will not increase under these maps.
Now describing the states with the high Hamming weight components removed by $\ket{\instalo{\either}}$, we have
\begin{align}
&\left\| \sum_{|{\bf b}|\le k'} \left[\ket{{\bf b}}\bra{{\bf b}}\otimes\trc \left(\measb{\either} \ket{\instalo{\either}}\bra{\instalo{\either}}\measb{\either}^\dagger\right) -  \ket{{\bf b}}\bra{{\bf b}}\otimes\trc \left(\measb{\either} \ket{\instate{\either}}\bra{\instate{\either}} \measb{\either}^\dagger\right)\right] \right\|_{\rm tr} \nonumber \\
&\le \left\| {\cal E}_\alpha(\ket{\instalo{\either}}\bra{\instalo{\either}}) - {\cal E}_\alpha(\ket{\instate{\either}}\bra{\instate{\either}}) \right\|_{\rm tr} \le \left\| \ket{\instalo{\either}}\bra{\instalo{\either}} - \ket{\instate{\either}}\bra{\instate{\either}} \right\|_{\rm tr} = O(\varepsilon) .
\end{align}
As a result, using the triangle inequality gives
\begin{equation}
\label{eq:erinhi}
D'_{\rm av} \le O(\varepsilon)+\left\| \sum_{|{\bf b}|\le k'} \left[\ket{{\bf b}}\bra{{\bf b}}\otimes\trc \left(\measb{\unco} \ket{\instalo{\unco}}\bra{\instalo{\unco}} \measb{\unco}^\dagger\right) -  \ket{{\bf b}}\bra{{\bf b}}\otimes\trc \left(\measb{\comp} \ket{\instalo{\comp}}\bra{\instalo{\comp}} \measb{\comp}^\dagger\right)\right] \right\|_{\rm tr}.
\end{equation}

Next, each measurement operator $\measb{\either}$ can be obtained by a sequence of measurement operators in our recursive measurement scheme, which will yield a sequence of measurement results $d_1, d_2, \ldots$.
Each ${\bf b}$ will correspond to a unique sequence of $d_\ell$ measurement results.
(Recall that $b_j$ are the individual results of measurements on uncompressed qubits, whereas $d_\ell$ are the individual results from the recursive measurement.)
Therefore we can relabel the basis states such that we have
\begin{equation}
\label{eq:dc}
D'_{\rm av} = \left\| \sum_{\bf d} \left\{\ket{{\bf d}}\bra{{\bf d}}\otimes\trc \left[\measc{\unco} \ket{\instate{\unco}}\bra{\instate{\unco}}(\measc{\unco})^\dagger\right] -  \ket{{\bf d}}\bra{{\bf d}}\otimes\trc \left[\measc{\comp} \ket{\instate{\comp}}\bra{\instate{\comp}}(\measc{\comp})^\dagger\right]\right\} \right\|_{\rm tr}.
\end{equation}
Now the measurement operators that are chosen at step $\ell$ in the recursive measurement scheme will depend on the measurement results that have been obtained at steps $1$ to $\ell-1$.
Therefore we can write the measurement operators as
\begin{align}
\measc{\either} &= \prod_{\ell=1}^{K} \meas{\either}{\ell}.
\end{align}
Here $K$ is the number of measurement operators to locate the ones.
For measurement result ${\bf b}$, the number of measurements required is no more than $1+2|{\bf b}|\log m$.
As we are taking ${\bf b}$ such that $|{\bf b}|\le k'$, we can take $K=1+2k'\log m$.

Using this notation, we can define CPTP maps by
\begin{equation}
\chan{\either}{\ell} (\rho) := \sum_{d_1,\ldots,d_\ell} \ket{d_\ell}\bra{d_\ell}\otimes \measp{\either}{\ell}\rho (\measp{\either}{\ell})^\dagger,
\end{equation}
where
\begin{equation}
\measp{\either}{\ell} := \ket{d_{\ell-1}}\bra{d_{\ell-1}} \otimes \cdots \otimes\ket{d_{1}}\bra{d_{1}} \otimes \meas{\either}{\ell}.
\end{equation}
Each map simply performs the appropriate measurement based on the prior measurement results (which are stored in ancillas), and appends an ancilla depending on the result of the measurement.
In term of these maps, the trace distance we wish to bound may be written as
\begin{equation}
D'_{\rm av} = \left\| \trc\chan{\unco}{K}\ldots\chan{\unco}{1}(  \ket{\instate{\unco}}\bra{\instate{\unco}}) -  \trc\chan{\comp}{K}\ldots\chan{\comp}{1} ( \ket{\instate{\comp}}\bra{\instate{\comp}}) \right\|_{\rm tr}.
\end{equation}

As has been noted above, we can omit the high Hamming weight contributions to the states $\ket{\instate{\either}}$, with a possible change in the trace distance of $O(\varepsilon)$.
The reason for this is that the trace distance is non-increasing under CPTP maps.
Our goal is now to successively approximate each of the maps in the sequence, at each stage bounding the introduced error by $O(\varepsilon)$.
At the end we will obtain two identical states, and then bound $D'_{\rm av}$ by $O(K\varepsilon)$.

More specifically, we want to approximate the evolution of the states for given measurement results as in Eqs.\ \eqref{eq:comap} and \eqref{eq:unap}.
Note that the reasoning given in the proof of Theorem~\ref{thm:errorfree} also gives a recursive method to determine the amplitudes in the states $\ket{\midstate{\either}{\ell}}$, starting from $\co_{x,0}^{()}=i^{|x|}\ca_{x}^m$.
This means that the definitions of these states are unambiguous.
We now consider the approximate unnormali{\s}ed states after $\ell-1$ measurements in the recursive measurement scheme $\ket{\midstate{\either}{\ell-1}}$, as given by Eqs.\ \eqref{eq:ent3} and \eqref{eq:ent4}.
We then define the states including the ancilla qubits containing the measurement results as
\begin{align}
\midrho{\either}{\ell-1} &:= \sum_{d_1,\ldots,d_{\ell-1}} \ket{d_{\ell-1}}\bra{d_{\ell-1}}\otimes\ldots\otimes\ket{d_1}\bra{d_1}\otimes \ket{\midstate{\either}{\ell-1}}  \bra{\midstate{\either}{\ell-1}}.
\end{align}
We wish to bound the error in approximating $\chan{\either}{\ell}(\midrho{\either}{\ell-1})$ by $\midrho{\either}{\ell}$.

In approximating $\chan{\unco}{\ell}(\midrho{\unco}{\ell-1})$ by $\midrho{\unco}{\ell}$ there is only one approximation: that of omitting the high Hamming weight states in applying the rotation.
The error in this approximation will be $O(\varepsilon)$ times the norm of the state.
Because the norm of the state is only changed by omitting high Hamming weight components, it can only be decreased.
Therefore the error is $O(\varepsilon)$.
Similarly, there is error in approximating $\chan{\comp}{\ell}(\midrho{\comp}{\ell-1})$ by $\midrho{\comp}{\ell}$ due to omitting high Hamming weight components, which is bounded by $O(\varepsilon)$.
There is also error because the $U_n$ rotations are not performed exactly.
Two such rotations are performed, each with error bounded by $O(\varepsilon)$, resulting in the overall error being bounded by $O(\varepsilon)$.

Therefore, we can start with Eq.\ \eqref{eq:dc}, remove the high Hamming weight components from the initial states, then proceed taking $\ell=1$ to $K$, replacing $\chan{\either}{\ell}(\midrho{\either}{\ell-1})$ by $\midrho{\either}{\ell}$ at each step.
At each step the distance is increased by $O(\varepsilon)$, and there are $K$ steps, so we obtain
\begin{equation}
D'_{\rm av} \le O(K\varepsilon) + \left\| \trc (\midrho{\unco}{K})-\trc ( \midrho{\comp}{K}) \right\|_{\rm tr}.
\end{equation}
But, because the same amplitudes have been obtained for the compressed and uncompressed cases, the same state is obtained after tracing over the control registers, and  $\trc (\midrho{\unco}{K})=\trc (\midrho{\comp}{K})$.
Therefore we obtain
\begin{equation}
D'_{\rm av} = O(K\varepsilon) = O(\varepsilon k'\log m).
\end{equation}
As $D_{\rm av}=O(\varepsilon'+D'_{\rm av})$, this yields Eq.\ \eqref{eq:erbnd}, as required.
\end{proof}

To summari{\s}e the sources of error in the above proof, these are as follows.
\begin{enumerate}
\item Omitting measurement results with Hamming weight greater than $k'$; see Eq.~\eqref{eq:erhigh}.
\item Omitting the high Hamming weight components of the initial states; see Eq.~\eqref{eq:erinhi}.
\item Omitting high Hamming weight components in each step of the recursive measurement.
\item Inaccuracy in performing the $U_n$ operations in each step of the recursive measurement.
\end{enumerate}

Error sources 3 and 4 give a contribution to the error of $O(\varepsilon)$ times the norm of the state for each step of the recursive measurement.
However, for many initial sequences of measurement results, at step $\ell$ all ones have already been located, so there are no further measurements needed.
This means that the measurements at this point are just the identity, and no further error is introduced for that sequence of initial measurement results.
This means that bounding the additional error by $O(\varepsilon)$ for each $\ell$ overestimates the error.
We will show that the error can be bound by using the mean number of ones that are measured.
In the case of the uncompressed measurements, the probability of each one is $\le 4\alpha^2\beta^2$.
Because $\beta^2\approx 1/8m$, the expected number of ones is $\le 4\beta^2 m=O(1)$.

\begin{theorem}[Improved error bounds]\label{thm:imper}
Provided $\varepsilon=O(1/(k'\log m))$, the error between the compressed and uncompressed schemes can be bounded as
\begin{equation}
D_{\rm av} = O(\varepsilon'+\varepsilon \log m).
\end{equation}
\end{theorem}
\begin{proof}
More specifically, $\midrho{\either}{\ell-1}$ will have a component where the recursive measurement scheme has not terminated yet, and another measurement needs to be performed.
This component will be that where the ancillas contain $d_1,\ldots,d_{\ell-1}$ corresponding to sequences of measurement results such that further measurements need to be performed.
There will also be a component corresponding to sequences of measurement results where the recursive measurement scheme has finished.
We will denote the components corresponding to that where the recursive measurement has not terminated or has terminated by $\midrho{\either}{\ell-1}^{\rm con}$ and $\midrho{\either}{\ell-1}^{\rm fin}$, respectively.
More explicitly, if we denote by $S_{\rm con}$ and $S_{\rm fin}$ the sets of measurement results $(d_1,\ldots,d_{\ell-1})$ that correspond to a recursive measurement that has not terminated or has terminated, respectively, then we have
\begin{align}
\midrho{\either}{\ell-1}^{\rm con/fin} &:= \sum_{(d_1,\ldots,d_{\ell-1})\in S_{\rm con/fin}} \ket{d_{\ell-1}}\bra{d_{\ell-1}}\otimes\ldots\otimes\ket{d_1}\bra{d_1}\otimes \ket{\midstate{\either}{\ell-1}}  \bra{\midstate{\either}{\ell-1}}.
\end{align}

Because the measurement acts only on $\midrho{\either}{\ell-1}^{\rm con}$, and the error in the measurement is bound by $O(\varepsilon)$ \emph{times} the trace of the state the measurement acts upon,
the error in approximating $\chan{\either}{\ell}(\midrho{\either}{\ell-1})$ by $\midrho{\either}{\ell}$ will be bounded by $O(\varepsilon\Tr(\midrho{\either}{\ell-1}^{\rm con}))$.
Therefore the total error from sources 3 and 4 is bounded by
\begin{equation}
O\left(\varepsilon\sum_{\ell=1}^K \Tr(\midrho{\either}{\ell-1}^{\rm con})\right).
\end{equation}
But, because the number of measurement steps need be no larger than $1+h\log m$, where $h$ is the number of ones found by the measurement, the probability that the number of ones is $\ge h$ is no greater than $\Tr(\midrho{\either}{h\log m}^{\rm con})$.
Denoting the probability that the number of ones is $\ge h$ by $p(|{\bf b}|\ge h)$, we have
\begin{equation}
\sum_{h=0}^{m} p(|{\bf b}|\ge h) =\sum_{h=0}^{m} \sum_{j=h}^m p(|{\bf b}|=j) =  \sum_{j=0}^m \sum_{h=0}^{j}p(|{\bf b}|=j) =  \sum_{j=0}^m (j+1) p(|{\bf b}|=j) =\langle |{\bf b}| \rangle +1.
\end{equation}
Therefore we can bound the sum of the traces by
\begin{align}
\sum_{\ell=1}^K \Tr(\midrho{\either}{\ell-1}^{\rm con}) &\le \sum_{\ell=1}^K \Tr(\midrho{\either}{\lfloor (\ell-1)/\log m \rfloor\log m}^{\rm con}) \nonumber \\
&\le \log m \sum_{h=0}^m  p(|{\bf b}|\ge h) = (\langle |{\bf b}| \rangle +1)\log m.
\end{align}

Next we need to take into account the fact that here the expectation value of the number of ones is for the approximate states $\midrho{\unco}{\ell-1}^{\rm con}$, not for the exact uncompressed measurement scheme.
To take account of this difference, we can use the cumulative error to bound the error in the norm of the state at each step.
Note that the norm of $\midrho{\either}{\ell-1}^{\rm con}$ is the same for $\either=\unco$ and $\either=\comp$, so we only need perform the analysis for $\either=\unco$.
Let $A_{\ell-1}$ denote the norm, for the exact uncompressed measurement, of the component where the recursive measurement scheme has not stopped before step $\ell$.
In addition, let $E_{\ell-1}$ denote the cumulative error before step $\ell$.
Then the increment in the error is bound by $\varepsilon$ times the norm of the non-terminated component, which is bound by $A_{\ell-1}$ plus the cumulative error.
\begin{align}
E_{\ell} &\le E_{\ell-1} + O(\varepsilon A_{\ell-1} + \varepsilon E_{\ell-1}) \nonumber \\
&= E_{\ell-1}[1+O(\varepsilon)] + O(\varepsilon A_{\ell-1}).
\end{align}
Multiplying both sides by $[1+O(\varepsilon)]^{K-\ell}$, we obtain
\begin{equation}
E_{\ell}[1+O(\varepsilon)]^{K-\ell} \le  E_{\ell-1}[1+O(\varepsilon)] ^{K-(\ell-1)}+ O(\varepsilon A_{\ell-1}[1+O(\varepsilon)]^{K-\ell}).
\end{equation}
As a result, the final error is bound by
\begin{align}
E_{K+1} &\le \sum_{\ell=1}^K O(\varepsilon A_{\ell-1}[1+O(\varepsilon)]^{K-\ell}) \nonumber \\
&\le [1+O(\varepsilon)]^K \sum_{\ell=1}^K O(\varepsilon A_{\ell-1}) \nonumber \\
&\le  O(\exp(\varepsilon K)\varepsilon(\langle|{\bf b}| \rangle +1)\log m).
\end{align}
Here the expectation value of the number of ones is for the exact scheme, which is $O(1)$, and we therefore find that the error is bound by $O(\exp(\varepsilon K)\varepsilon\log m)$.
Recall that we take $K=O(k'\log m)$.
This means that, provided $\varepsilon = O(1/(k'\log m))$, $\exp(\varepsilon K)$ is $O(1)$, and we obtain scaling of the error of $O(\varepsilon\log m)$.
Adding $O(\varepsilon'+\varepsilon)$ to take account of error sources 1 and 2 yields the result given in the Theorem.
\end{proof}

Given the conditions of this Theorem, the overall error for each time step is $O(\varepsilon'+\varepsilon\log m)$.
This includes error in simulating the driving Hamiltonian.
The driving Hamiltonian may be applied up to $k'$ times, though the expected number of times is $O(1)$.
As the allowable error in the driving Hamiltonian is $O(\varepsilon')$, that gives a contribution of $O(\varepsilon')$ to the error in each time step.
As there are $O(T)$ time steps, the total error is $O(\varepsilon' T+\varepsilon T\log m)$.
To limit the error of the overall scheme to $\varepsilon_{\rm tot}$, we take $\varepsilon'=O(\varepsilon_{\rm tot}/T)$ and $\varepsilon=O(\varepsilon_{\rm tot}/(T\log m))$.
Then $k'=O(\log(1/\varepsilon'))=O(\log(T/\varepsilon_{\rm tot}))$.
As we consider large $T$ and small $\varepsilon_{\rm tot}$, we therefore have $\varepsilon=O(1/[\log(T/\varepsilon_{\rm tot})\log m])=O(1/(k'\log m))$.
This means that the condition of the Theorem is satisfied with this choice of parameters, and the total error will be bounded by $\varepsilon_{\rm tot}$.

\section{Proof of Main Theorem}
\label{sec:proof}

Finally we are in a position to prove Theorem \ref{th:main}.

\begin{proof}{\bf of Theorem \ref{th:main}.}
First, the number of oracle queries is $O(k'T)$, because we have divided the simulation into $O(T)$ time intervals, and limit the number of queries required within each time interval to $O(k')$.
The value of $k'$ is chosen to ensure that the error due to omitting high Hamming-weight states $O(1)$ times within each time interval is no more than $\varepsilon'$.
We can bound the total error by $\varepsilon_{\rm tot}$ if we take $\varepsilon'=O(\varepsilon_{\rm tot}/T)$, which means that $k'$ scales as
\begin{equation}
\label{eq:smaller}
k' = O\left( \frac{\log(T/\varepsilon_{\rm tot})}{\log\log(T/\varepsilon_{\rm tot})} \right).
\end{equation}
Then the overall number of oracle calls scales as
\begin{equation}
O\left( \frac{T\log(T/\varepsilon_{\rm tot})}{\log\log(T/\varepsilon_{\rm tot})} \right).
\end{equation}
Omitting the dependence on $\varepsilon_{\rm tot}$ gives the result given in the statement of the Theorem.

Next we discuss the number of gates required for Alg.~\ref{alg:measure}.
The maximum number of steps in the recursive procedure is $1+2k'\log m$, but the expected number of steps is $O(\log m)$.
For the full algorithm for the evolution over time $T$, there are many of these recursive measurements, and the probability of the average number of steps differing significantly from its expected value is small.
Similarly to the analysis in Section \ref{sec:over}, an upper bound of $O(1/\varepsilon_{\rm tot})$ times the average value will not be exceeded with probability $1-O(\varepsilon_{\rm tot})$.
As $\varepsilon_{\rm tot}$ is taken to be constant, this does not affect the final result.
Because $U_n$ and $U_n^\dagger$ are performed at each step, these operations are performed $O(\log m)$ times.
As was found above, the complexity of the operation $U_n$ is $O(k[\log m+\log\log(1/\varepsilon)])$.
Therefore the overall complexity for this time step is $O(k[(\log m)^2+\log m\log\log(1/\varepsilon)])$.

It is also necessary to perform $O(k')$ time evolutions under the driving Hamiltonian.
In the definition of the problem we let $G$ be the number of gates required for the simulation of the driving Hamiltonian, so that the number of gates to simulate the driving Hamiltonian in this time step is $O(k'G)$.
Therefore, the scaling for the total number of gates is
\begin{equation}
O\left( TGk' + Tk [(\log m)^2+\log m\log\log(1/\varepsilon)] \right).
\end{equation}

Next we use $\varepsilon'=O(\varepsilon_{\rm tot}/T)$ and $\varepsilon=O(\varepsilon_{\rm tot}/(T\log m))$.
As discussed in Section \ref{sec:m}, we can take $\log m=O(\log(\|H\|T/\varepsilon_{\rm tot}))$.
Considering the scaling with large $\|H\|$, the total number of gates simplifies to
\begin{equation}
O\left( \frac{TG\log(T/\varepsilon_{\rm tot})}{\log\log(T/\varepsilon_{\rm tot})}
+\frac{T\log[(T\log m)/\varepsilon_{\rm tot}]}{\log\log[(T\log m)/\varepsilon_{\rm tot}]}(\log m)^2
\right).
\end{equation}
A further simplification may be obtained by ignoring the double-log factors in the denominators, and then using the scaling of $\log m$ to give
\begin{equation}
O\left( TG\log(T/\varepsilon_{\rm tot})
+T[\log(T/\varepsilon_{\rm tot})+\log\log\|H\|][\log(T/\varepsilon_{\rm tot})+\log\|H\|]^2
\right).
\end{equation}
The number of gates can then be bounded in a simpler but looser form as
\begin{equation}
O\left( TG\log(T/\varepsilon_{\rm tot})
+T[\log(\|H\|T/\varepsilon_{\rm tot})]^3
\right).
\end{equation}
Omitting $\varepsilon_{\rm tot}$, because we take this quantity to be constant, gives the scaling in the Theorem.

The number of qubits required for the algorithm is dominated by the number of qubits required for the recursive measurement scheme.
The number of qubits used for the ancilla space is $O(k[\log m+\log\log(1/\varepsilon)])$.
In the recursive measurement scheme it may be necessary to duplicate the ancilla space $k'$ times to ensure that a maximum of $k'$ ones are detected.
The overall space used is therefore
\begin{equation}
O\left(\frac{\log[(T\log m)/\varepsilon_{\rm tot}]}{\log\log[(T\log m)/\varepsilon_{\rm tot}]}\frac{\log(T/\varepsilon_{\rm tot})}{\log\log(T/\varepsilon_{\rm tot})}[\log m+\log\log(T\log m/\varepsilon_{\rm tot})]\right).
\end{equation}
Cancelling the double-log, then omitting double-log factors in the denominator gives
\begin{equation}
O\left(\log[(T\log m)/\varepsilon_{\rm tot}]\log(T/\varepsilon_{\rm tot})\log m\right).
\end{equation}
Using the scaling of $m$ then gives
\begin{equation}
O\left(\log(T/\varepsilon_{\rm tot})[\log(T/\varepsilon_{\rm tot})+\log\log\|H\|]\log(\|H\|T/\varepsilon_{\rm tot})\right).
\end{equation}
A simpler bound can be given as
\begin{equation}
O\left([\log(\|H\|T/\varepsilon_{\rm tot})]^3\right).
\end{equation}
Again omitting $\varepsilon_{\rm tot}$ gives the scaling in the statement of the Theorem.

Note also that the allowable error in the driving Hamiltonian is $O(\varepsilon')$, which is $O(\varepsilon_{\rm tot}/T)$.
For constant $\varepsilon_{\rm tot}$, the allowable error in the implementation of the driving Hamiltonian is $O(1/T)$, as given in the statement of the Theorem.
\end{proof}

\section{Conclusions}
We have shown that any continuous-time query algorithm of cost $T$ can be implemented with a number of discrete queries close to linear in $T$, and with a number of gates that is also close to linear in $T$.
This means that any continuous-time quantum algorithm can be converted into an efficient discrete-query algorithm.
In contrast, using the algorithm of Ref.~\cite{CleveG+2009} directly would result in a number of gates that is linear in $mT$.
That is, the gate complexity would be superlinear in $\|H\|T$, and similar to what would be obtained just using product formulae.

Our results provide an even better improvement in the scaling with $\|H\|$; the number of gates is polylogarithmic in this quantity, rather than superlinear.
As the norm of the driving Hamiltonian can potentially be very large, this can potentially provide a very large improvement in efficiency.
In both cases, the query complexity is independent of $\|H\|$, but it does not appear to be possible to completely remove the dependence of the number of gates on $\|H\|$ via this approach.

The methods we have presented may also be used as an alternative to product formulae when simulating state evolution for a sum of Hamiltonians, where one Hamiltonian is self-inverse, and the other has large norm, $\|H\|$.
Previous work has considered the complexity of Hamiltonian simulation via product formulae where one Hamiltonian has much larger norm \cite{Papa2012}.
Even using that approach, the complexity is only reduced from $O(\|H\|T(\|H\|T/\varepsilon_{\rm tot})^{\delta})$ to $O(\|H\|T(T/\varepsilon_{\rm tot})^{\delta})$.
In comparison, here we have obtained complexity that is polylogarithmic in $\|H\|$.

\end{document}